\documentclass[a4paper,UKenglish,cleveref,autoref,thm-restate]{lipics-v2021}

\pdfoutput=1 \hideLIPIcs

\usepackage{xspace}
\usepackage[disable]{todonotes}
\usepackage{cleveref}
\usepackage{mathtools}

\theoremstyle{plain}

\Crefname{redrule}{Solution Normalization}{Solution Normalization}

\usepackage{etoolbox}
\newtoggle{long}\toggletrue{long} \iftoggle{long}{
	\NewDocumentEnvironment{prooflater}{m}{\expandafter\global\expandafter\def\csname#1\endcsname{}
    \begin{proof}}{\end{proof}}
    
	\NewDocumentEnvironment{proofsketch}{o +b}{}{}
	\newcommand{\restateref}[1]{}
    \newcommand{\restaterefName}[2]{[#2]}
}{
	\NewDocumentEnvironment{prooflater}{m +b}{ \expandafter\global\expandafter\def\csname#1\endcsname{\begin{proof}#2\end{proof}}}{}
	\NewDocumentEnvironment{proofsketch}{O{Proof sketch.}}{\begin{proof}[#1]}{\end{proof}}
	\usepackage{apptools}
\newcommand{\restateref}[1]{[\IfAppendix{\hyperref[#1]{$\star$}}{\hyperref[#1*]{$\star$}}]}
    \newcommand{\restaterefName}[2]{[#2, \IfAppendix{\hyperref[#1]{$\star$}}{\hyperref[#1*]{$\star$}}]}
}

\let\oldrestatable\restatable
\def\restatable{\expandafter\oldrestatable}

\newcommand{\cplan}{{\normalfont\textsc{Clus}\-\textsc{tered} \textsc{Pla}\-\textsc{nari}\-\textsc{ty}}\xspace}
\newcommand{\lplan}{{\normalfont\textsc{Level} \textsc{Pla}\-\textsc{nari}\-\textsc{ty}}\xspace}

\newcommand{\yclplan}{{\normalfont($y$-)\textsc{mono}\-\textsc{tone} \textsc{Clus}\-\textsc{tered} \textsc{Level} \textsc{Pla}\-\textsc{nari}\-\textsc{ty}}\xspace}
\newcommand{\cclplan}{{\normalfont\textsc{Con}\-\textsc{vex} \textsc{Clus}\-\textsc{tered} \textsc{Level} \textsc{Pla}\-\textsc{nari}\-\textsc{ty}}\xspace}

\newcommand{\shortyclplan}{{\normalfont\textsc{mCLP}}\xspace}
\newcommand{\shortcclplan}{{\normalfont\textsc{cCLP}}\xspace}

\def\yclp{\shortyclplan}
\def\cclp{\shortcclplan}

\DeclareMathOperator{\att}{att}

\DeclareMathOperator{\fr}{fr}
\DeclareMathOperator{\cl}{cl}
\DeclareMathOperator{\vc}{vc}

\graphicspath{{./graphics/}, {./graphics/pdf}}

\title{Monotone Clustered Level Planarity}

\author{Simon D. Fink}{Algorithms and Complexity Group, Technische Universität Wien, Austria}{sfink@ac.tuwien.ac.at}{https://orcid.org/0000-0002-2754-1195}{}

\author{Matthias Pfretzschner}{Faculty of Computer Science and Mathematics, University of Passau, Germany}{pfretzschner@fim.uni-passau.de}{https://orcid.org/0000-0002-5378-1694}{}

\author{Ignaz Rutter}{Faculty of Computer Science and Mathematics, University of Passau, Germany}{rutter@fim.uni-passau.de}{https://orcid.org/0000-0002-3794-4406}{}

\author{Marie Diana Sieper}{Institute of Computer Science, University of Würzburg, Germany}{marie.sieper@uni-wuerzburg.de}{https://orcid.org/0009-0003-7491-2811}{}

\authorrunning{S.D. Fink, M. Pfretzscher, I. Rutter, M.D. Sieper} 

\Copyright{Simon D. Fink, Matthias Pfretzscher, Ignaz Rutter, and Marie Diana Sieper} 



\relatedversion{}

\nolinenumbers

\EventEditors{Maarten L\"{o}ffler and Silvia Miksch}
\EventNoEds{2}
\EventLongTitle{34th International Symposium on Graph Drawing and Network Visualization (GD 2026)}
\EventShortTitle{GD 2026}
\EventAcronym{GD}
\EventYear{2026}
\EventDate{August 19--21, 2026}
\EventLocation{Ontario, Canada}
\EventLogo{}
\SeriesVolume{396}
\ArticleNo{12}

\begin{document}

\maketitle

\begin{abstract}
  We consider the combination of the two constrained planarity problems \textsc{Level-} and \cplan.
  Traditionally, level-planar drawings with convex clusters have been studied in this setting.
  Fink et al.~\cite{fink_clustered_2024} recently introduced a different way of combining level- and clustered planarity by mimicking a classic characterization of clustered planarity in the level-planar setting:
  The problem \yclplan (\shortyclplan) seeks a level-planar drawing in which it is possible to augment each cluster with edges that do not cross cluster boundaries so that it becomes connected while maintaining level-planarity.
  This is in line with previous research on clustered planarity that poses certain requirements on the augmentation edges that make each cluster connected, e.g., that they form a path.
  Fink et al.~\cite{fink_clustered_2024} showed that \yclp is NP-complete even for biconnected single-source graphs and instances with a constant number of levels and clusters.

  We further classify the parameterized complexity of the \yclp problem by, on the one hand, showing hardness even for instances that consist of a forest with trees of bounded size, no isolated vertices, and a small constant number of either clusters or levels.
  This excludes fixed-parameter tractability for almost all graph-structural parameters, except for vertex cover, even in conjunction with the number of clusters.
  We complement this by showing fixed-parameter tractability when parameterizing by the vertex cover number and the number of clusters.
  A major obstacle that we overcome is the fact that \yclp is non-hereditary, i.e., subinstances of yes-instances may be no-instances and vice versa, which makes it challenging to apply usual reduction techniques.

  \iftoggle{long}{}{
  \subparagraph{Generative AI Declaration}
  Generative AI was not used in the preparation of this article.
  }
\end{abstract}

\section{Introduction}
In the constrained planarity setting, we ask whether a given graph can be drawn in the plane such that no two edges cross and some given additional constraints are satisfied.
For example, in \lplan, vertices need to lie on predefined horizontal lines, and in \cplan, given sets of vertices shall be drawn in closed regions that may not overlap.
Such problems can be used to find visualizations that also include further real-world meta-information like hierarchy levels in organizational charts, which may further include non-hierarchical teams and working groups.
Another example are UML diagrams of software source code, where directories or packages induce a hierarchy and, additionally, the assignment to individual software features can provide a grouping.
The constrained planarity problems modelling such applications thus received much attention in the graph drawing community, see e.g.\ \cite{dloz-pgw-15,sch-tat-13} for an overview.

\subparagraph{Level Planarity.}
Formally, the input for the problem \lplan is a \emph{level graph} $(G,\ell)$, that is a graph $G=(V,E)$ and a \emph{leveling function}~$\ell \colon V \to [k]=\{1,2,\ldots,k\}$ with $k\in\mathbb{N}$ that assigns vertices to \emph{levels} such that no two adjacent vertices are on the same level.
We seek a \emph{level planar drawing} of~$(G,\ell)$, that is a crossing-free drawing of~$G$ that positions each vertex~$v$ on the corresponding horizontal line~$\ell(v)$ and draws each edge as a strictly $y$-monotone curve between its endpoints.
A linear-time algorithm for testing \lplan, that is finding such level planar drawing, is known since the 1990s~\cite{jlm-lpt-98}.
We call a level graph $(G,\ell)$ \emph{proper} if $|\ell(u) - \ell(v)| = 1$ for every edge $uv\in E$.

\subparagraph{Clustered Planarity.}
A \emph{flat clustered graph} is a pair $(G, T)$ where $T \colon V(G) \to \mathbb{N}$ is a \emph{flat clustering} of $G$, a (not necessarily injective) function that assigns a natural number to each vertex.
Note that a flat clustering induces a partition of $V(G)$.
A \emph{cluster} of~$(G, T)$ is a natural number $\mu \in \mathbb{N}$ with $T^{-1}(\mu) \neq \emptyset$ and for each $v \in V(G)$, we refer to $T(v)$ as the \emph{cluster of $v$}.
We let $|T|$ denote the number of clusters in $(G, T)$, i.e., $|T| = |\{T(v) \mid v \in V(G)\}|$.

A \emph{clustered planar drawing} of~$(G,T)$ is a planar drawing of~$G$ that assigns each cluster~$\mu$ to a region~$R$ bounded by a simple closed curve such that (i)~$R$~contains exactly the vertices $T^{-1}(\mu)$, (ii)~no two region boundaries intersect, and (iii)~no edge crosses a region boundary more than once.
Especially, this means that an edge crosses a cluster boundary if and only if exactly one endpoint lies inside the cluster.  
A clustered graph is \emph{clustered planar} if it admits such a drawing, and the corresponding decision problem is called \cplan.
We call a clustered graph \emph{cluster-connected} if, for every cluster $\mu$,  $T^{-1}(\mu)$ induces a connected subgraph in~$G$.
Initial solutions to \cplan handled only cluster-connected instances~\cite{fce-htd-95,len-hpt-89}, while the general case remained open for over 30 years.
Among many other results obtained in the meantime (see e.g.\ \cite{dalozzo_exact_2024,di_battista_efficient_2009} for listings), it has been shown that allowing a \emph{non-flat clustering}, that is a laminar family of sets instead of a partition, reduces to the flat case and thus does not increase complexity~\cite{cp-cp-18}.
Furthermore, Cornelsen and Wagner gave an alternative classification showing that a planar drawing~$\mathcal{E}$ of~$G$ is clustered planar with regard to a given clustering if and only if we can insert a set of \emph{augmentation edges} (sometimes also called saturators) into~$\mathcal{E}$ planarly such that $G$ becomes cluster-connected and, for every cluster $\mu$, no cycle induced by~$T^{-1}(\mu)$ encloses a vertex not in~$T^{-1}(\mu)$~\cite{cw-ccc-06} (see also~\cite[Theorem 1]{di_battista_efficient_2009}).
Recently, Fulek and Tóth~\cite{ft-aec-22} presented the first general polynomial-time algorithm, which was soon after improved to quadratic time~\cite{bfr-spw-21}.

\subparagraph{Clustered Level Planarity.}
In this paper, we consider the combination of the problems \textsc{Level-} and \cplan, that is we are given both a leveling and a clustering and seek a \emph{cl-planar} drawing that at the same time satisfies the constraints modelled by both problems; see \Cref{fig:non-cclplan} for an example.
This combination has initially been studied in 2004 by Forster and Bachmaier \cite{fb-clp-04}, who posed the additional requirement of cluster boundaries being convex in the resulting drawing, which is why we refer to this variant as \cclplan (\cclp).
They showed that \cclp is solvable in linear time if the instance is proper 
and \emph{level-connected} (each cluster contains an edge between any pair of adjacent levels it spans).
Angelini et al.~\cite{alb-tio-15} gave a quadratic algorithm for the non-level-connected but proper case and show that the non-proper case is NP-complete.
Fink et al.~\cite{fink_clustered_2024} recently studied the problem without the convexity constraint and gave a polynomial-time algorithm for biconnected graphs with a \emph{single source}, i.e., a single vertex that does not have a neighbor on a lower level.
\begin{figure}[t]
  \begin{subfigure}{.5\textwidth}
    \centering
    \includegraphics[page=1]{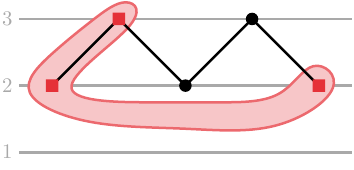}
    \subcaption{}
  \end{subfigure}\begin{subfigure}{.5\textwidth}
    \centering
    \includegraphics[page=2]{graphics/pdf/mclp-example}
    \subcaption{}
  \end{subfigure}
  \caption{
    \textbf{(a)} A drawing that is cl-planar but not monotone-cl-planar. 
    The red vertices represent one cluster, the black vertices another cluster. 
    The underlying $\yclp$-instance is a no-instance, because the rightmost vertex cannot be connected to one of the other two vertices of the red cluster using a $y$-monotone augmentation edge.
    \textbf{(b)} After inserting the vertex~$v$ of the red cluster on level~$1$, the instance becomes a yes-instance of \yclp, as illustrated by the dashed red augmentation edges. 
}
  \label{fig:non-cclplan}
\end{figure}

\subparagraph{Monotone Clustered Level Planarity.}
Fink et al.~\cite{fink_clustered_2024} additionally considered a problem variant of \cclp that replaces the convexity requirement for clusters with a monotonicity requirement.
The problem \yclplan (\shortyclplan) is based on the Cornelsen and Wagner~\cite{cw-ccc-06} classification of cluster planar drawings described above.
A level-planar drawing is called \emph{monotone-cl-planar}~\cite{fink_clustered_2024} if we can insert a set of augmentation edges such that (1) the graph becomes cluster-connected, (2) for every cluster~$\mu$, no cycle induced by $T^{-1}(\mu)$ encloses a vertex not in $T^{-1}(\mu)$, and (3) the drawing remains level-planar.
Note that this requires the augmentation edges to also be drawn $y$-monotone.
The difference to the edges of $G$ is that augmentation edges may also be drawn horizontally, that is, connect consecutive vertices on the same level.
This is in line with other research posing restrictions on the augmentation edges in a cluster planar drawing, e.g., requiring them to form a path~\cite{dalozzo_exact_2024}.
Fink et al.\ \cite{fink_clustered_2024} already showed that \yclp is NP-hard even for biconnected 
single-source graphs and for instances with at most five levels and at most two clusters\footnote{In the non-flat setting their ``single non-trivial cluster'' that is distinct from the root cluster encompassing all vertices is equivalent to having two clusters in our flat clustered setting.}.

\subparagraph{Contributions and Structure.}

\begin{figure}[t]
	\centering
	\includegraphics[page=1]{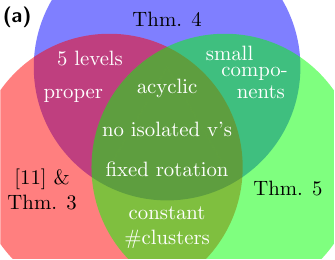}
    \caption{
	Overview of the combinations of properties for NP-hard input instances.
    }
    \label{tab:hardness-prop-combs}
\end{figure}

We build on this latter result and in \Cref{sec:complexity} show that the problem is also hard if all connected components are trees of constant size, the instance contains no isolated vertices, and either the number of levels or the number of clusters is bounded by a small constant.
\Cref{tab:hardness-prop-combs} gives an overview of (combinations of) parameters that we show hardness for.
Note that these results imply para-NP-hardness for the majority of common graph-structural parameters, even in conjunction with the number of clusters or the number of levels.
This essentially only leaves the \emph{vertex cover number} 
(the minimum number of vertices whose deletion removes all edges) 
as an interesting graph-structural parameter.

While we do not settle the tractability of the vertex cover number on its own, we show in \Cref{sec:fpt} that \yclp is FPT with respect to the vertex cover number plus the number of clusters.
Here, the biggest challenge is the fact that \yclp is non-hereditary (see \Cref{fig:non-cclplan}).
This makes it particularly difficult to formulate safe reduction rules for kernelization, as removing a vertex of the graph can turn a yes-instance into a no-instance and vice versa.
Therefore, we instead show fixed-parameter tractability using a different approach.

Instead of reducing the size of an input instance via kernelization, we work with so-called solution normalizations, which reduce the size of a given solution.
This way we can show that a \yclp-instance is a yes-instance if and only if it has a subinstance of bounded size that admits a (special) solution.
While this immediately implies an XP-algorithm for \yclp, we cannot enumerate all such subinstances in FPT time.
Instead, we consider so-called blueprints, which are \yclp instances that coincide with the original instance only on an important ``core'' subgraph.
All other vertices of the blueprint are dummy-vertices with only very abstract leveling information.
We can then test whether a given blueprint can be realized by associating these dummy-vertices with vertices of the input graph using dynamic programming.
Since the number of blueprints is bounded, we can enumerate all of them in FPT time and report a yes-instance if and only if one of them is realizable.
Together, our results significantly extend our understanding of the parameterized complexity of \yclp.
\iftoggle{long}{}{
The full proofs of results marked with a (clickable) star are deferred to the appendix.
}

\section{Preliminaries}
\label{sec:preliminaries}
Let $\mathcal I =(G, \ell, T)$ be a clustered level graph with leveling function $\ell$ and flat clustering~$T$.
A \emph{monotone-cl-planar embedding} of $\mathcal I$ is an equivalence class of monotone-cl-planar drawings of $\mathcal I$, where two drawings are equivalent if and only if the order of vertices and edge-level crossings, including the augmentation edges, coincides in both drawings on every level.
Note that it is easy to obtain a corresponding monotone-cl-planar drawing from a monotone-cl-planar embedding and vice versa.
We refer to a monotone-cl-planar drawing $\Gamma$ of $\mathcal I$ as a pair $(\mathcal E, \mathcal A)$, where $\mathcal E$ is the restriction of $\Gamma$ to $G$ and $\mathcal A$ is the restriction of $\Gamma$ to the augmentation edges. 
For two vertices $u, v \in V(G)$ with $T(u) = T(v) = \mu$, a \emph{cluster-connection between $u$ and $v$} is a $uv$-path in $\mathcal E \cup \mathcal A = \Gamma$ that consists of vertices that belong to $\mu$.

For $v \in V(G)$, we let $N(v) = \{u \in V(G) \mid \{u,v\} \in E(G)\}$ denote the \emph{open neighborhood} of~$v$ in $G$, i.e., the neighbors of $v$ in $G$ (excluding $v$ itself).
For a path $P$ and two vertices $u, v$ of $P$, we let $P[u,v]$ denote the subpath of $P$ that starts at $u$ and ends at $v$.
We need the following two elementary observations.

\begin{restatable}\restateref{prop:sequences}{proposition}{propSequences}
    \label{prop:sequences}
    Let $\phi = (v_1, \dots, v_n)$ be a sequence of $n \geq 3$ distinct elements and let $C \subseteq \{(v_i, v_{i+1}) \mid 1 \leq i < n\}$ be a set of consecutive pairs in $\phi$.
    If $n \geq 2 |C| + 3$, there exists a contiguous subsequence $(a, b, c)$ of $\phi$ such that $(a,b), (b,c) \notin C$.
\end{restatable}
\begin{prooflater}{proofSequences}
    Note that $\phi$ contains exactly $n-2$ contiguous subsequences of length 3.
    Since each pair of $C$ is contained in at most two of these contiguous subsequences, $n - 2 > 2|C|$ implies that at least one contiguous subsequence of length 3 contains no element of $C$.
\end{prooflater}
\begin{restatable}\restateref{prop:polygon-diagonal}{proposition}{propPolygonDiagonal}
\label{prop:polygon-diagonal}
    Let $K$ be a simple $n$-gon with $n \geq 5$ and let $v$ be a vertex of $K$.
    Then there exist two non-adjacent vertices $x, y \neq v$ of $K$ such that a $y$-monotone curve between $x$ and $y$ can be drawn in the closure $\overline{K}$ of $K$ (the union of its interior with its boundary).
\end{restatable}
\begin{prooflater}{proofPolygonDiagonal}
    Consider an arbitrary triangulation of $K$.
    If this triangulation contains a diagonal of $K$ that is not incident to $v$, we are immediately done.
    We thus assume that all $n-3$ diagonals of the triangulation are incident to $v$.

    Consider the 5-gon $L$ induced by vertices $v, x_1, x_2, x_3, x_4$ that are consecutive on the boundary of $K$.
    Since all diagonals of the triangulation are incident to $v$, observe that $L$ is contained in $K$.

    Let $y(u)$ denote the $y$-coordinate of a vertex $u$.
    If $y(x_{i-1}) \leq y(x_{i}) \leq y(x_{i+1})$ or $y(x_{i-1}) \geq y(x_{i}) \geq y(x_{i+1})$ holds for some $i \in \{2,3\}$, the union of $x_{i-1}x_i$ and $x_ix_{i+1}$ is the desired $y$-monotone curve.
    We thus now assume that $x_1, x_2, x_3, x_4$ form a zigzag.
    
    First consider the case where $y(x_i) \leq y(v)$ for all $i \in \{1,2,3,4\}$; the case where $y(x_i) \geq y(v)$ holds is symmetric.
    Since $x_1, x_2, x_3, x_4$ form a zigzag, exactly one of the interior angles at $x_2$ and $x_3$ is convex.
    Without loss of generality, we assume this holds for $x_2$.
    Since $y(x_1) \leq y(v)$, $y(x_2) \leq y(v)$ and $y(x_3) \leq y(v)$, $v$ does not lie inside the triangle $T$ formed by $x_1, x_2$ and $x_3$.
    The vertex $x_4$ also does not lie inside $T$, since otherwise the triangulation edge $vx_3$ would not lie inside $K$, a contradiction.
    Hence $T$ is fully contained in $L$ and thus the straight line segment $x_1x_3$ is the desired $y$-monotone curve.

    Finally, consider the case where some of $x_1, x_2, x_3, x_4$ lie above $v$ and others lie below~$v$.
    No matter how the four vertices are distributed, there always exists a non-adjacent pair $x_i \neq x_j$ with $y(x_i) \leq y(v) \leq y(x_j)$.
    Since the line segments $x_iv$ and $vx_j$ either lie on the boundary of $K$ or are diagonals of the triangulation, the union of $x_iv$ and $vx_j$ yields the desired $y$-monotone curve.
\end{prooflater}

We assume basic familiarity with parameterized complexity (see, e.g.,~\cite{Cygan15} for an overview).

\section{Hardness of \texorpdfstring{$y$}{y}-monotone Clustered Level Planarity}
\label{sec:complexity}

\begin{figure}[p]
\begin{subfigure}{\textwidth}
	\centering
	\includegraphics[page=2]{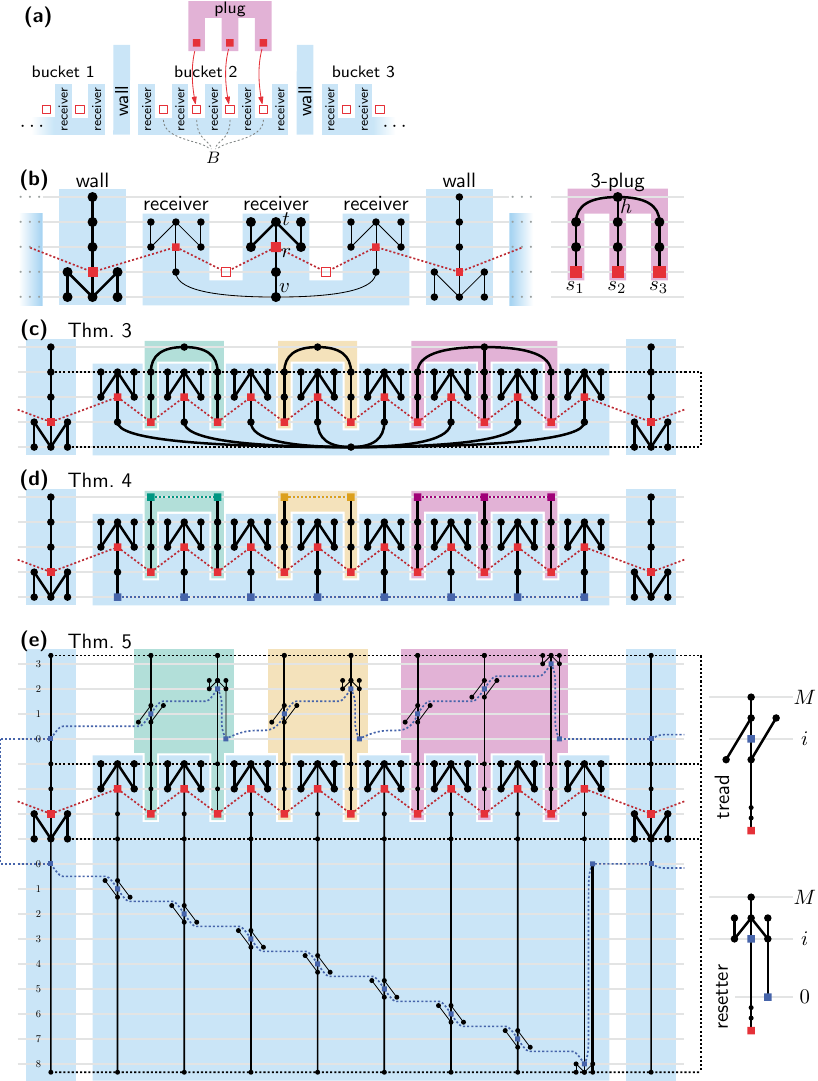}
    \phantomsubcaption\label{fig:3partition}
	\phantomsubcaption\label{fig:3part-gadgets}
	\phantomsubcaption\label{fig:3part-red1}
	\phantomsubcaption\label{fig:3part-red2}
	\phantomsubcaption\label{fig:3part-red3}
\end{subfigure}
\caption{
    \textbf{(a)} Sketch of a \textsc{3-Partition} instance with 3 buckets, each requiring $B = 4$ red squares to be filled, and a 3-plug providing 3 squares, which cannot be split over multiple buckets due to walls. 
	\textbf{(b)} The \yclp gadgets we use to model walls, receivers, and plugs.
    The colors of the vertices correspond to the clusters and the background colors highlight plugs, buckets, and walls.
  	Dashed edges represent the $y$-monotone augmentation for the respective cluster, where the red augmentation requires further vertices of the red cluster filling the red boxes.
  \textbf{(c)} A bucket of size 7, filled with two 2-plugs and a 3-plug.
  	As there are exactly enough plug pins to fill the spaces between receivers, walls need to be evenly distributed between the buckets to allow for augmentation.
\textbf{(d)} An equivalent instance where high-degree vertices are replaced by clusters ensuring that vertices
  belonging to the same plug/bucket stay together.
\textbf{(e)} An equivalent instance where high-degree vertices are replaced by stair tread and resetter gadgets that ensure the right (number of) components stay together.
}
\label{fig:3part-reds}
\end{figure}

The authors in~\cite{fink_clustered_2024} show that \yclp is NP-complete, even with only two (flat) clusters and spanning only five levels, by a reduction from the \textsc{3-Partition} problem, which asks to partition a set of numbers into buckets, each summing to a given value $B$ (see \Cref{fig:3partition}).

\begin{restatable}\restaterefName{thm:threepart-levelcluster}{{\cite[Theorem 3.2]{fink_clustered_2024}}}{theorem}{threepartthm}
  \label{thm:threepart-levelcluster}
  \yclp is NP-complete, even if the input is proper and acyclic, contains at most two clusters and five levels, and all vertices have a fixed rotation.
\end{restatable}
\begin{prooflater}{threepartproof}
	We only sketch the correctness of the slightly modified construction here as it also follows analogously to the proof given by \cite{fink_clustered_2024}.
	The reduction is from the strongly NP-complete problem \textsc{3-Partition}~\cite{GareyJ79}, which takes as input a number $m \in \mathbb{N}$, a set $A = \{a_1, \dots, a_{3m}\}$ of positive integers, and a number $B \in \mathbb{N^+}$ such that $B/4 < a_i < B/2$ and $\sum_{a \in A} a = m \cdot B$.
	The problem asks whether there exists a partition of $A$ into $m$ sets such that the sum of each set equals $B$.

	Given an instance $(m, A, B)$ of \textsc{3-Partition}, we build $m$ \emph{buckets of size $B$} that are obtained by chaining together $B + 1$ \emph{receivers} at a common vertex $v$; see \Cref{fig:3partition,fig:3part-gadgets}.
	We additionally construct $m-1$ \emph{walls} that span all five levels and whose purpose is to separate the $m$ buckets.
	Each positive integer $a_i \in A$ is represented by an \emph{$a_i$-plug} consisting of $a_i$ paths that join in a single vertex $h$ on the topmost level; see \Cref{fig:3part-gadgets}.
	The bottommost vertex of each of these paths is called a \emph{pin}.

	The instance contains a cluster $\mu$ that consists of the pins of all plugs and a single vertex of each receiver and each wall (colored red in \Cref{fig:3part-gadgets,fig:3part-red1}).
	The latter are chosen in a way that guarantees that each receiver can only have $y$-monotone augmentation edges to a single pin or wall to its left and to a single pin or wall to its right.
	The other cluster (colored black in \Cref{fig:3part-red1}) can be augmented horizontally on the lowest and second-highest levels (shown as dashed lines) and by connecting both augmentation paths to the right of the last wall.
	Since the $\mu$ augmentation prevents a wall from lying between two receivers of the same bucket, each bucket requires at least $B$ pins to allow $y$-monotone connections between its cluster vertices.
	However, since the instance contains~$m \cdot B$ pins overall, it follows that each bucket must contain precisely $B$ pins.
	Moreover, since the outermost receivers of two buckets cannot be augmented directly to each other and all $m \cdot B$ pins must lie inside some bucket, each pair of adjacent buckets is separated by a wall.
	This guarantees that all pins of each plug lie in the same bucket; see \Cref{fig:3part-red1}.

	The resulting \yclp instance is a yes-instance if and only if the plugs can be partitioned into $m$ sets of size $B$, which is the case if and only if $(m, A, B)$ is a yes-instance of \textsc{3-Partition}.
\end{prooflater}

\Cref{fig:3part-gadgets,fig:3part-red1} sketch a slightly modified construction, which allows us to extend the reduction and show that hardness also holds if all components of the graph have constant size:
By replacing high-degree vertices with degree-1 vertices on the same level that belong to a common cluster to ensure their consecutivity (see \Cref{fig:3part-red2}), we obtain an equivalent instance with an unbounded number of clusters, but constant-size connected components.
\begin{restatable}\restateref{thm:threepart-levelvi}{theorem}{thmConstCompLevel}
  \label{thm:threepart-levelvi}
  \yclp is NP-complete, even if the input is proper and acyclic, contains at most five levels, every  component has constant size, and all vertices have a fixed~rotation.
\end{restatable}
\begin{prooflater}{proofConstCompLevel}
  Let $(m, A, B)$ be an instance of \textsc{3-Partition} and let $(G, \ell, T)$ be the equivalent instance of \yclplan obtained from the reduction in \Cref{thm:threepart-levelcluster}.
  In every plug $P$ of $G$, we first replace the vertex $h$ (see \Cref{fig:3part-gadgets}) with $\deg(h)$ degree-1 vertices that are incident to the original neighbors of $h$, thus turning $P$ into $\deg(h)$ paths of length 4.
  To ensure that these paths remain consecutive, the new degree-1 vertices that replace $h$ are placed in a new cluster; see \Cref{fig:3part-red2}.
  As there a no other vertices in the new cluster, the paths belonging to $P$ can then be connected with $y$-monotone (or actually strictly horizontal) curves if and only if no walls or paths of other plugs lie between them.
  Observe that this modification preserves any solution of $(G, \ell, T)$, and contracting the horizontal augmentation edges of the new clusters in a solution yields a solution to the original instance $(G, \ell, T)$.
  Thus, this modification yields an equivalent instance.

  Analogously applying this procedure to the high-degree vertex of every bucket on the lowest level gives an instance where every connected component contains a constant number of vertices.
  To avoid problems with augmenting the old black cluster (see \Cref{fig:3part-red1}), we exploit the already unbounded number of clusters and put every black vertex into its own individual cluster.
  This yields an equivalent instance as we can easily insert a horizontal augmentation when undoing this change and again putting all affected vertices into a single black cluster.
\end{prooflater}

Using a different construction to ensure that the right number of components for each plug and each bucket remain together (see \Cref{fig:3part-red3}), we
further bound the number of clusters
at the cost of introducing an unbounded number of levels and making the instance~non-proper.
\begin{restatable}\restateref{thm:threepart-levelstairs}{theorem}{thmLevelstairs}
  \label{thm:threepart-levelstairs}
  \yclp is NP-complete, even if the input is acyclic, contains at most three clusters, every component has constant size, and all vertices have a fixed rotation.
\end{restatable}
\begin{prooflater}{proofLevelstairs}
  Let $(m, A, B)$ be an instance of \textsc{3-Partition} and let $(G, \ell, T)$ be the equivalent instance of \yclplan obtained from the reduction in \Cref{thm:threepart-levelcluster}.
  We construct an equivalent instance $(G', \ell', T')$ with only three clusters and where every connected component has constant size as follows; see \Cref{fig:3part-red3}.
  The interaction between pins, receivers, and walls remains the same as in \Cref{thm:threepart-levelcluster,thm:threepart-levelvi}, but since the reduction in \Cref{thm:threepart-levelvi} requires many clusters, we need a different mechanism that ensures that pins of plugs and receivers of buckets remain consecutive.
  In the following, we first encode the restriction that pins of the same plug must be consecutive with respect to all walls and the pins of other plugs.
  Let $M < B/2$ be the maximum integer in $A$.
  To make the following description less verbose, we use the notion of \emph{layers} next to levels.
  For each $i \in \{0, \dots, M\}$, layer $i$ refers to level $l(i) = 5 + 3i$, which ensures that gadgets we place on different layers lie above previously-used levels and have enough space between them to achieve the desired interaction.

  Let $P$ be a plug of $G$ and let $s_1, \dots, s_a$ denote the individual pins of $P$.
  Similar to the reduction in \Cref{thm:threepart-levelvi}, we first remove the high-degree vertex $h$ of $P$ (see \Cref{fig:3part-gadgets}) and replace it with $\deg(h)$ degree-1 vertices on the topmost level, each of which is adjacent to one of the original neighbors of $h$.
  Additionally, we attach a topmost-level vertex to the top of each wall.
  Therefore, every pin and every wall is now attached to a path that extends to the topmost level.
  Roughly speaking, these paths will additionally contain vertices of a common cluster and structures on different layers that use the $y$-monotonicity of the cluster connections to enforce a stair-like arrangement; see \Cref{fig:3part-red3}.
  This fixes the order of the pins of $P$ and prevents walls or pins of other plugs from being interspersed.
  We call these new cluster-vertices the \emph{blue vertices}, the other vertices are either \emph{red} or \emph{black} (see \Cref{fig:3part-red3}).

  More specifically, for each $i \in \{1, \dots, a - 1\}$ in an $a$-plug $P$, the path containing $s_i$ is attached to a \emph{stair tread} on layer $i$.
  It has the structure shown in \Cref{fig:3part-red3} on the right and contains a blue vertex on layer $i$ that can only connect to vertices on higher layers on one side, and only to vertices on lower layers on the other side.
  We call the two sides the \emph{lower side} and \emph{higher side}, respectively.
  
  The pin $s_a$ is designated as the \emph{resetter}; also shown in \Cref{fig:3part-red3} on the right.
  It has a fixed embedding up to horizontal mirroring and contains a blue vertex on layer $a$ that can only connect to vertices on lower layers on one side (the \emph{lower side}).
  On the other side, it contains a blue vertex on layer $0$.
  Additionally, each wall contains a blue subdivision vertex on layer $0$, see \Cref{fig:3part-red3}.

  We repeat the same construction (mirrored vertically, extending below the previously lowest level) for every bucket by replacing the unique high-degree vertex $v$ of every bucket analogously; see \Cref{fig:3part-gadgets,fig:3part-red3}.
  Note that $(G', \ell', T')$ now consists of three clusters, since it additionally contains a single cluster consisting of the blue vertices compared to $(G, \ell, T)$.
  Moreover, every connected component now has constant size.
  It remains to show that $(G, \ell, T)$ and $(G', \ell', T')$ are equivalent.

  First assume that there exists a \yclp solution $(\mathcal E, \mathcal A)$ of $(G, \ell, T)$.
  We obtain a drawing $(\mathcal E', \mathcal A')$ of $(G', \ell', T')$ by replacing the high-degree vertices of plugs and buckets with paths as described above.
  Since the $a$ pins of a plug $P$ are indistinguishable in $G$, we can reorder them such that they appear in the sequence $s_1, \dots, s_a$ from left to right in $G'$. 
  For $i \in \{1, \dots a-1\}$, the pin $s_i$ is attached to a stair tread on layer $i$ and its higher side can thus connect directly to the lower side of its adjacent stair tread (or resetter) $s_{i + 1}$ using a $y$-monotone curve.
  The right side of the resetter $s_a$ can connect to the lower side of the stair tread $s'_1$ of the plug $P'$ to its right.
  If there is a wall to its right, then $s_a$ can instead connect to the blue vertex on layer $0$ of the wall, which in turn connects to the stair tread $s'_1$ of the following plug.
  The $y$-monotone connections of blue vertices that ensure the consecutivity of pins thus now form a path that we call the \emph{upper path}.
  Applying the same procedure for the receivers, we obtain an additional \emph{lower path} connecting the blue vertices that ensure the consecutivity for the receivers.
  Recall from the proof of \Cref{thm:threepart-levelcluster} that every wall lies between two plugs.
  The left- and rightmost blue vertices of the upper and lower paths must therefore belong to pins.
  As these extremal vertices can always connect beyond their respective layer 0, the two leftmost vertices can connect via a $y$-monotone curve and the cluster consisting of all blue vertices is thus now connected, but still separates the black cluster.
  To make the black cluster connected again, we add horizontal augmentation edges to the top- and bottommost levels, which we connect with the previous augmentation for the black cluster at the very right (compare \Cref{fig:3part-red1,fig:3part-red3}).
  Note that the augmentation for the red cluster remains unaffected and thus all three clusters can be made connected.

  Now assume that there exists monotone-cl-planar drawing  $(\mathcal E', \mathcal A')$ of $(G', \ell', T')$.
  Recall that the structure of the walls and the red vertices ensure that no two walls lie next to each other.
  To show that there exists a monotone-cl-planar drawing  $(\mathcal E, \mathcal A)$ of $(G, \ell, T)$, it thus suffices to show that the pins of plugs and the receivers of buckets are consecutive in $(\mathcal E', \mathcal A')$.
  Again, we only show that this holds for the pins, the restriction for the receivers can be shown analogously.
  
  Note that, for each $i \in \{1, \dots, M\}$, there are exactly $R_i \coloneqq |\{a \in A \mid a = i\}|$ resetters and $S_i \coloneqq |\{a \in A \mid a > i\}|$ stair treads on layer $i$ overall.
  Observe that $R_i + S_i = S_{i-1}$ for $i \in \{2, \dots, M\}$.
  Since all pins and walls extend from layer 0 to layer $M$, note that each stair tread and resetter can only be connected to a single stair tread or resetter on each side in  $(\mathcal E', \mathcal A')$.
  We show via induction that the higher side of every stair tread $T$ placed on layer $i$ for $i \in \{1, \dots, M - 1\}$ must connect to the lower side of a stair tread or a resetter on layer $i + 1$ in $(\mathcal E', \mathcal A')$.

  As a base case, consider the $S_{M-1}$ stair treads on layer $M - 1$.
  Recall that, on its higher side, a stair tread can only connect to blue vertices that lie above.
  Since only layer $M$ lies above layer $M-1$ and since $S_M = 0$ and $R_M = S_{M-1}$, the higher side of each of the $S_{M-1}$ stairs on layer $M-1$ must connect to the lower side of the $R_M$ resetters on layer $M$.
  
  Now assume that the statement holds for $\{i + 1, \dots, M - 1\}$ and consider layer $i$.
  Recall the higher side of one stair tread can only connect to levels strictly above, and conversely must connect to the lower side of the above stair tread or resetter.
  Assume for the sake of contradiction that the higher side of a stair tread $T$ on layer $i$ connects to the lower side of a stair tread or resetter on layer $j \geq i + 2$.
  By the inductive hypothesis, each of the $S_{j - 1}$ stair treads on layer $j - 1$ connects to the lower side of a stair tread or resetter on layer $j$.
  Since $S_{j - 1} = S_j + R_j$ and each side has at most one augmentation edge, the lower side of every stair tread or resetter on layer $j$ is thus already ``occupied'' by a stair tread on layer $j - 1$.
  Therefore, no stair tread or resetter on layer $j$ can connect to $T$, a contradiction.
  The higher side of $T$ must therefore be connected to a stair tread or resetter on layer $i + 1$, which concludes the inductive proof.
  
  Note that this implies that for every $a \in A$, there exists a consecutive sequence of stair treads on layers $1, \dots, a - 1$, followed by a resetter on layer $a$.
  We can thus contract all paths of these pins into a single vertex to obtain the corresponding plug of $G$.
  Applying the same procedure for the buckets, we find a $y$-monotone cl-planar drawing $(\mathcal E, \mathcal A)$ of $(G, \ell, T)$.
\end{prooflater}

\section{\yclp is FPT parameterized by \texorpdfstring{$\vc(G)+|T|$}{vc(G)+|T|}}
\label{sec:fpt}
In this section, we show that \yclp is FPT parameterized by the size of a vertex cover plus the number of clusters.
Our approach rests on the well-known fact that the number of vertices of degree at least~3 in a planar graph is bounded linearly in its vertex cover number~\cite{fomin_kernelization_2019}.
However, as \yclp is non-hereditary, employing the standard kernelization approach is difficult.
We circumvent this as follows.
Based on the vertex cover and clusters, we obtain a ``core'' graph of bounded size, at first disregarding the augmentability to monotone-cl-planar.
We then show that an instance of \yclp is a yes-instance if and only if it has a bounded-size subgraph that admits a (special) solution and contains this core graph.
In other words, for yes-instances, there exists a bounded-size set of vertices outside of the core that enables augmentability to monotone-cl-planar.
As we cannot enumerate all such sets in FPT time, we instead consider so-called blueprints where these vertices are present as dummy-vertices with only very abstract leveling information.
We can then test whether a given blueprint can be realized by associating these dummy-vertices with vertices of the input graph using dynamic programming.
Since the number of such blueprints is bounded, we can enumerate all of them in FPT time and report a yes-instance if and only if one of them is realizable.

\begin{figure}
  \begin{subfigure}{.32\textwidth}
    \centering
    \includegraphics[page=1]{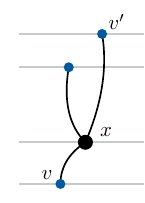}
    \subcaption{}
    \label{fig:equivalenceA}
  \end{subfigure}
  \begin{subfigure}{.32\textwidth}
    \centering
    \includegraphics[page=2]{graphics/pdf/equivalence.pdf}
    \subcaption{}
    \label{fig:equivalenceB}
  \end{subfigure}
  \begin{subfigure}{.32\textwidth}
    \centering
    \includegraphics[page=3]{graphics/pdf/equivalence.pdf}
    \subcaption{}
    \label{fig:equivalenceC}
  \end{subfigure}
  \caption{\textbf{(a)} Three $x$-pendant vertices that belong to the blue cluster, i.e., the three vertices are equivalent. The upper representative of $v$ is $v'$ and $v$ is its own lower representative. 
  \textbf{(b)} Equivalent degree-2 vertices, where $a$ is an upper $xy$-ear, $b$ is an $xy$-transversal, and $c$ is a lower~$xy$-ear.
  \textbf{(c)} Two $xy$-ears $a, b$ that lie on the same level $l$. Note that a third $xy$-ear on level $l$ immediately implies a no-instance, as it cannot connect to $y$ and $x$ with crossing-free $y$-monotone edges.
  }
  \label{fig:equivalence}
\end{figure}

\subparagraph*{Notation.}
Let $\mathcal I=(G, \ell, T)$ be an instance of \yclp and let $Q$ be a vertex cover of $G$.
Two vertices $u,v\in V(G)\setminus Q$ are \emph{equivalent} if $N(u) = N(v)$ and $T(u) = T(v)$, that is, they have the same neighborhood and cluster.
To break ties, we use an arbitrary linear ordering~$<$ of~$V(G)$ such that $u < v$ implies~$\ell(u) \le \ell(v)$.
For $v \in V(G) \setminus Q$, we define the \emph{lower representative} and \emph{upper representative} of $v$ as the (with respect to~$<$) smallest and largest vertex equivalent to~$v$, respectively.
We refer to a \mbox{degree-1} vertex $v \in V(G)\setminus Q$ with neighbor $x\in Q$ as \emph{($x$-)pendant} and to a degree-2 vertex $v\in V(G) \setminus Q$ with neighbors $x<y\in Q$ as \emph{\mbox{($xy$-)transversal}} if $\ell(x)<\ell(v)<\ell(y)$ and as \emph{($xy$-)ear} otherwise. 
We further call $v$ a \emph{lower $xy$-ear} if~$\ell(v) < \min\{\ell(x),\ell(y)\}$ and an \emph{upper $xy$-ear} if~$\ell(v) > \max \{\ell(x),\ell(y)\}$;~see~\Cref{fig:equivalence}.

Our algorithm rests on the following central observation; see also \Cref{fig:equivalenceC}.

\begin{observation}\label{obs:no-three-equivalent-same-ell}
    Any instance of \yclp with three or more ears on the same level and with the same neighborhood is a no-instance.
\end{observation}

\subparagraph*{The Core.}
We define $C_0$ as the set of all vertices in the vertex cover $Q$, all vertices of degree at least~3, and the lower and upper representative of all vertices of degree at most~1.
Since the number of vertices of degree at least~3 is bounded linearly in~$|Q|$~\cite[Lemma~13.3]{fomin_kernelization_2019} and there are $O(|Q| \cdot |T|)$ representatives for the vertices of degree at most 1, $|C_0|$ is in $O(|Q| \cdot |T|)$.
The \emph{core vertex set}~$C$ consists of $C_0$ and, for every $u, v \in Q$ and every cluster~$\mu$ of $(G, T)$, $|C_0| + 2$ $uv$-transversal vertices in cluster $\mu$  (or all of them, if there are fewer).
We define the \emph{core graph} $G[C]$ of $\mathcal I$ as the graph induced by $C$.
Note that the size of $C$ is in $O(|Q|^3 \cdot |T|^2)$.

For an $\yclp$-solution $(\mathcal E, \mathcal A)$ and a vertex $v \in V(G) \setminus C$, we refer to the degree of $v$ in~$\mathcal A$ (that is, the number of its incident augmentation edges) as the \emph{aug-degree} of $v$.
We say that edges of $\mathcal E$ (i.e., of $G$) are \emph{real edges}, while edges of $\mathcal A$ are \emph{augmentation edges}.

\subsection{Subinstances and Tidy Solutions}

Let $\mathcal I=(G, \ell, T)$ be an instance of \yclp with core vertex set $C$.
A set $X \subseteq V(G)$ is a \emph{subinstance of $\mathcal I$} if $C \subseteq X$. For a subinstance~$X$, we define~$\mathcal I[X] = (G[X],\ell\vert_{X}, T\vert_{X})$.  Let $X \subseteq V(G)$ be a subinstance and let $(\mathcal E, \mathcal A)$ be a solution of $\mathcal I[X]$.
Let~$v \in V(G) \setminus X$ be an ear with neighbors~$x,y \in C$.   
In order to be able to reinsert $v$ into $(\mathcal E, \mathcal A)$, we require specific augmentation edges in $(\mathcal E, \mathcal A)$.
We say that the ear~$v$ is \emph{covered by an augmentation edge $e = ab\in \mathcal A$} if 
\begin{itemize}
    \item $v$, $a$, and $b$ are all upper $xy$-ears or all lower $xy$-ears with $\ell(a) < \ell(v) < \ell(b)$,
    \item the two faces incident to $e$ are triangles, and
    \item $v$, $a$, and $b$ are equivalent or $\{T(v), T(a), T(b)\} \subseteq \{T(x), T(y)\}$.
\end{itemize}
We will show that the first two properties allow us to reinsert $v$ along $e$ into $(\mathcal E, \mathcal A)$ without introducing crossings, while the last property then helps us maintain cluster connectivity.
A solution $(\mathcal E, \mathcal A)$ of $\mathcal I[X]$ is \emph{tidy} if 
every ear in $V(G) \setminus X$ is covered by an augmentation edge and if every pair of ears that lie on the same level and have the same neighborhood is covered by two distinct augmentation edges.  
Note that, for the sake of these covering edges, we explicitly allow augmentation edges between vertices that belong to different clusters.

\begin{figure}
  \begin{subfigure}{.5\textwidth}
    \centering
    \includegraphics[page=2]{graphics/pdf/reinsert}
    \subcaption{}
    \label{fig:reinsertA}
  \end{subfigure}
  \begin{subfigure}{.5\textwidth}
    \centering
    \includegraphics[page=1]{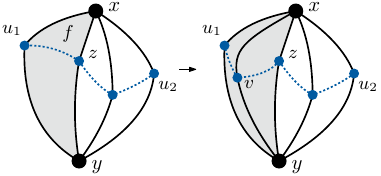}
    \subcaption{}
    \label{fig:reinsertB}
  \end{subfigure}
    \begin{subfigure}{.5\textwidth}
    \centering
    \includegraphics[page=7]{graphics/pdf/reinsert}
    \subcaption{}
    \label{fig:reinsertCA}
  \end{subfigure}
  \begin{subfigure}{.5\textwidth}
    \centering
    \includegraphics[page=3]{graphics/pdf/reinsert}
    \subcaption{}
    \label{fig:reinsertC}
  \end{subfigure}
   \begin{subfigure}{\textwidth}
    \centering
    \includegraphics[page=4]{graphics/pdf/reinsert}
    \subcaption{}
    \label{fig:reinsertD}
  \end{subfigure}
  \begin{subfigure}{.49\textwidth}
    \centering
    \includegraphics[page=5]{graphics/pdf/reinsert}
    \subcaption{}
    \label{fig:reinsertE}
  \end{subfigure}
  \begin{subfigure}{.49\textwidth}
    \centering
    \includegraphics[page=6]{graphics/pdf/reinsert}
    \subcaption{}
    \label{fig:reinsertF}
  \end{subfigure}
  \caption{ Reinserting a vertex $v$ if $v$ is an $xy$-transversal \textbf{(a)}-\textbf{(b)}, if $v$ is an $xy$-ear \textbf{(c)}-\textbf{(d)}, if $v$ is a degree-1 vertex \textbf{(e)}, and if $v$ is an isolated vertex, whose upper and lower representative must be connected by some path $P$ \textbf{(f)}-\textbf{(g)}. Vertices of different colors belong to different clusters and large vertices belong to the core. Dashed edges represent augmentation edges.
  }
  \label{fig:reinsert}
\end{figure}
\begin{restatable}\restateref{lm:tidySubinstance}{lemma}{lmTidySubinstance}
\label{lm:tidySubinstance}
  Let $\mathcal I=(G, \ell, T)$ be an instance of \yclp.
  Then $\mathcal I$ is a yes-instance if and only if it has a subinstance $X$ such that $\mathcal I[X]$ admits a tidy solution.
\end{restatable}
\begin{proofsketch}
    If $\mathcal I$ is a yes-instance, then $V(G)$ is a subinstance of $\mathcal I$ that admits a solution, which is tidy as there are no ears not in $V(G)$.
    For the other direction, we can always use vertices of the core in a tidy solution $(\mathcal E, \mathcal A)$ along which we can reinsert a given vertex $v \in V(G) \setminus X$; see \Cref{fig:reinsert}.
    If $v$ is an $xy$-transversal, we use the fact that we kept $|C_0| + 2$ $xy$-transversals equivalent to $v$ in the core to find a cycle $f$ in whose interior we reinsert~$v$ (Figures~\ref{fig:reinsertA}--\ref{fig:reinsertB}).
For an $xy$-ear $v$, the tidiness of the solution guarantees the existence of an edge $ab$ that covers $v$ and we can reinsert $v$ as a subdivision vertex on~$ab$ while maintaining cluster connectivity (Figures~\ref{fig:reinsertCA}--\ref{fig:reinsertC}).
For pendant and isolated vertices, we reinsert $v$ guided by the upper representative~$v'$ and/or the lower representative~$v''$ of~$v$ (\Cref{fig:reinsertD,fig:reinsertE,fig:reinsertF}).
\end{proofsketch}
\begin{prooflater}{proofTidySubinstance}
  If $\mathcal I$ is a yes-instance, then $V(G)$ is a subinstance of $\mathcal I$ that admits a solution, which is tidy as there are no ears not in $V(G)$.
  For the other direction, consider a subinstance $X$ of $\mathcal I$ such that $\mathcal I[X]$ admits a tidy solution $(\mathcal E, \mathcal A)$.
  We show how to iteratively reinsert all vertices of $V(G)\setminus X$, which all have degree at most 2. To this end, we first reinsert all vertices of degree 2 while preserving tidiness of the solution.

  Let $v \in V(G) \setminus X$ with $\deg(v)=2$, let $\mu = T(v)$, and let $x, y \in C$ denote the two neighbors of $v$, where without loss of generality~$\ell(x) \le \ell(y)$.
  We distinguish between two cases, based on the position of $\ell(v)$ with respect to $\ell(x)$ and $\ell(y)$.
  
  If $\ell(x) \le \ell(v) \le \ell(y)$, since there are no edges between vertices on the same level, we have~$\ell(x) < \ell(v) < \ell(y)$, i.e., we are in the transversal case.
  Since~$v \notin C$, we retained at least $|C_0| + 2$ other $xy$-transversal vertices $U\subset C$ equivalent to~$v$.
  Consequently,~$\mathcal E$ contains a cycle $K$ formed by $x,y$ and two vertices $u_1,u_2\in U$ such that the interior of~$K$ contains no vertices of $C_0$.  Let~$f$ denote the face that lies in the interior of~$K$ and is incident to~$u_1$.   Since we have not yet reinserted pendant and isolated vertices of $V(G) \setminus X$, all vertices contained in the interior of $K$ have degree~$2$ and are adjacent to~$x$ and~$y$. It follows that~$f = xu_1yz$ for some~$z$ has size~$4$ in~$\mathcal E$. Note that $z = u_2$ is possible. If $f$ does not contain the augmentation edge~$u_1z$, we can insert~$v$ parallel to~$xu_1y$ along with the augmentation edge~$u_1v$; see \Cref{fig:reinsertA}.  Otherwise, we remove the augmentation edge~$u_1z$, insert $v$ parallel to~$xu_1y$, and insert the augmentation edges $u_1v,vz$; see \Cref{fig:reinsertB}.
  Since no ears are involved, the solution remains tidy.
  
  Otherwise, $v$ is an $xy$-ear, and we assume without loss of generality that it is an upper~$xy$-ear, i.e.,~$\ell(x) \le \ell(y) < \ell(v)$.  The case of a lower~$xy$-ear, where~$\ell(v) < \ell(x) \le \ell(y)$, is symmetric.
  As the solution $(\mathcal E, \mathcal A)$ for $\mathcal I[X]$ is tidy, there is an augmentation edge $e = ab$ in~$\mathcal A$ that covers~$v$.  We then insert $v$ as a subdivision vertex on $e$ and route its edges to $x,y$ first along $e$ and then following the edges of $a$ to $x$ and $y$, respectively, to obtain a solution for $\mathcal I' = \mathcal I[X \cup \{v\}]$; see \Cref{fig:reinsertCA,fig:reinsertC}. Note that this yields a crossing-free $y$-monotone drawing, since the fact that~$e$ covers~$v$ implies that $\ell(x) \le \ell(y) < \ell(a) < \ell(v) < \ell(b)$ and the two faces incident to $e$ are triangles (and therefore empty).
  Moreover, if~$a,b,v$ are all equivalent, the subdivision of the augmentation edge guarantees that their cluster remains connected; see \Cref{fig:reinsertCA}.  Otherwise, we have~$\{T(a),T(b),T(v)\} \subseteq \{T(x),T(y)\}$ and their clusters trivially remain connected with $x$ and $y$; see \Cref{fig:reinsertC}.  Hence, we have obtained a solution~$(\mathcal E',\mathcal A')$ of~$\mathcal I'$.
  
  Observe that in this solution, we split the edge $e=ab$ into two edges $e'=av$ and $e''=bv$, thus we need to ensure that all ears not in $X$ are still covered for $\mathcal I'$ to be tidy.  First, note that all ears that are covered by an edge distinct from~$e$ are still covered by that same edge.
Any ear $v'$ with $\ell(v')<\ell(v)$ that was covered by $e$ is now covered by $e'$, while those with $\ell(v')>\ell(v)$ are now covered by $e''$.
  This leaves $xy$-ears $v' \neq v$ with $\ell(v)=\ell(v')$.
  By \Cref{obs:no-three-equivalent-same-ell} we can assume that there are no two such vertices in $\mathcal I'$, as otherwise $\mathcal I$ would contain three $xy$-ears on level $\ell(v)$, implying that it is a no-instance.
  As the definition of tidiness requires two distinct covering augmentation edges $e,f$ in case of two $xy$-ears on the same level, this leaves $v'$ covered by $f$, and we thus always obtain a tidy solution.

  By repeating these steps, we eventually reinsert all degree-2 vertices.
  Since this includes all ears, observe that every solution of the instance is now trivially tidy.
  We now insert the degree-1 vertices.
  To this end, let $v \in V(G) \setminus X$ with $\deg(v) = 1$ and let $x \in C$ denote the unique neighbor of $v$.
  Without loss of generality, we assume $\ell(x) < \ell(v)$.
  Let $v' \in C$ denote the upper representative of $v$, i.e., $\ell(v) \leq \ell(v')$.
  First assume $\ell(v) < \ell(v')$.
  Observe that we can add a $y$-monotone copy of the edge $xv'$ to $(\mathcal E, \mathcal A)$ without introducing crossings by routing it directly next to the original.
  Since $\ell(x) < \ell(v) < \ell(v')$, we can add $v$ as a subdivision vertex on level $\ell(v)$ and turn the edge $vv'$ into an augmentation edge for cluster $\mu = T(v) = T(v')$; see \Cref{fig:reinsertD} (left).
  Observe that $xv$ and $vv'$ are $y$-monotone and that the cluster $T(v) = T(v')$ is connected in the resulting drawing, i.e., we have a solution.
  If $\ell(v) = \ell(v')$, we place $v$ directly to the right of $v'$, route the edge $xv$ directly to the right of $xv'$, and add the horizontal augmentation edge $vv'$ for cluster $\mu$; see \Cref{fig:reinsertD} (right).
  Since $\ell(v) = \ell(v')$, there may be augmentation edges incident to $v'$ that cross the new edge $xv$, no matter how close we place $v$ to $v'$, namely those that leave~$v'$ downwards to the right or horizontally to the right.
  In this case, we simply reattach these augmentation edges from $v'$ to $v$; see \Cref{fig:reinsertD} (right).
  Since $v$ and $v'$ are connected by an augmentation edge, $\mu$ remains connected.

  Finally, it remains to insert isolated vertices.
  Let $v \in V(G) \setminus X$ with $\deg(v) = 0$ and let $v', v'' \in C$ denote the upper and lower representative of $v$, respectively and thus $T(v) = T(v') = T(v'') =: \mu$.
  Since $\mu$ must be connected in the solution, there exists a path $P$ between $v'$ and $v''$ in $(\mathcal E \cup \mathcal A)[T^{-1}(\mu)]$.
  Because $\ell(v') \geq \ell(v) \geq \ell(v'')$, there exists a point in which $P$ intersects level $\ell(v)$.
  Note that such a point is either a vertex or the interior point of an (augmentation) edge.
  
  If there exists a vertex $p$ of $P$ on level $\ell(v)$, we insert $v$ directly to the left of $p$, reattach any horizontal augmentation edge that enters $p$ from the left (if it exists) to $v$ and add the horizontal augmentation edge $vp$; see \Cref{fig:reinsertE}.

  Otherwise, if no vertex of $P$ lies on level $\ell(v)$, there is a point $p$ on level $\ell(v)$ that is an interior point of an (augmentation) edge $e$ such that no endpoint of $e$ lies on $\ell(v)$.
  Due to $y$-monotonicity, we may assume without loss of generality that $p$ is the only interior point of~$e$ on level $\ell(v)$.
  If $e$ is an augmentation edge, we can simply place $v$ as a subdivision vertex of $e$ on $p$.
  If $e$ is a real edge, let $x$ denote the endpoint of $e$ that lies on a higher level than~$\ell(v)$.
  We place $v$ directly to the left of $p$ and route a new augmentation edge $vx$ directly to the left of the segment of $e$ between $p$ and $x$; see \Cref{fig:reinsertF}.
  Hence $\mu$ remains connected and since $e$ is $y$-monotone, $vx$ is also $y$-monotone.
\end{prooflater}

Our goal is to show that every yes-instance has a bounded-size subinstance that admits a tidy solution.  To this end, we show that, similarly to reduction rules for kernelization, we can remove certain non-core vertices while preserving the property that the resulting instance admits a tidy solution.  We refer to this process as \emph{solution normalization}.  Note that, unlike reduction rules for kernelization, our normalization rules exploit a tidy solution, which we for now assume to be known.
We will first bound the number of degree-2 vertices. 
Recall that every degree-2 vertex that is not contained in the core is either a transversal or an ear.
We start by reducing the former.

\begin{figure}
  \begin{subfigure}{.5\textwidth}
    \centering
    \includegraphics[page=2, scale=0.9]{graphics/pdf/Rule2.pdf}
    \subcaption{}
    \label{fig:transversalA}
  \end{subfigure}
  \begin{subfigure}{.5\textwidth}
    \centering
    \includegraphics[page=1, scale=0.9]{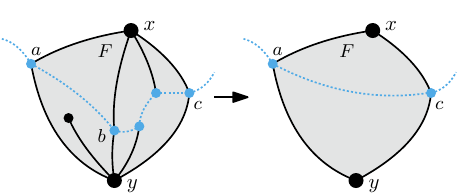}
    \subcaption{}
    \label{fig:transversalB}
  \end{subfigure}
  \caption{Illustration of the two cases of \Cref{rr:transversal}.}
  \label{fig:transversal}
\end{figure}

\begin{lemma}[Solution Normalization 1]
    \label{rr:transversal}
    Let $X$ be a subinstance of $\mathcal I$ that has a tidy solution~$(\mathcal E, \mathcal A)$.
    If $X \setminus C$ contains at least $2|C|+3$ equivalent $xy$-transversals with $\{x,y\} \subseteq C$, there exists a subinstance $X' \subsetneq X$ of $\mathcal I$ that has a tidy solution.
\end{lemma}
\begin{proof}Since $X \setminus C$ contains at least $2|C|+3$ equivalent $xy$-transversals, by \Cref{prop:sequences} there are three such transversals $a, b, c$ such that the cycles $axby$ and $bxcy$ have disjoint interiors and neither of these interiors contains a vertex of $C$; see \Cref{fig:transversal}.
    Let $F$ denote the region that is bounded by the cycle~$axcy$.  Note that~$b$ is contained in~$F$ and moreover,~$F$ does not contain any ear.
    Let $X'$ be the subinstance obtained from~$X$ by removing all vertices that lie inside~$F$.  We claim that~$X'$ admits a tidy solution.  To show this, we distinguish between two cases as follows.
    
    If $x$ and $y$ belong to the same cluster $\mu$ and the removal disconnects $\mu$ (see \Cref{fig:transversalA}), then~$\mu$ is different from $T(a) = T(c)$ and $x$ and $y$ are cluster-connected inside~$F$ in $(\mathcal E, \mathcal A)$.
    Hence $a$ and $c$ cannot be cluster-connected inside $F$, since otherwise the two cluster-connections would cross inside $F$.
    We can thus simply add the augmentation edge $xy$ and all clusters remain connected.

    Otherwise, the only cluster that can be disconnected by removing all vertices from $F$ is the cluster $T(a) = T(c)$. In this case, since $a$ and~$c$ are~$xy$-transversals, we can insert the augmentation edge $ac$ and all clusters remain connected; see \Cref{fig:transversalB}.

    Note that in both cases the solution remains tidy, since no ears are affected.
\end{proof}

With the following two lemmas, we reduce the number of ears in solutions.
We start with equivalent upper (lower) ears that lie in a cluster other than their two neighbors in the core.
Subsequently, we use this solution normalization to bound the number of ears with identical neighborhood.
Both normalizations use a similar approach as \Cref{rr:transversal}; see~\Cref{fig:ears}.

\begin{figure}
        \begin{subfigure}{.5\textwidth}
        \centering
        \includegraphics[page=1, scale=0.9]{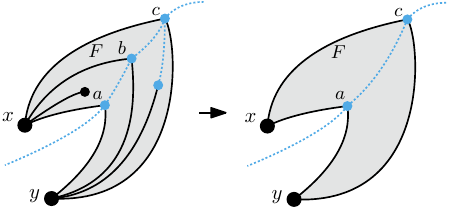}
        \subcaption{}
        \label{fig:earA}
      \end{subfigure}
      \begin{subfigure}{.5\textwidth}
        \centering
        \includegraphics[page=2, scale=0.9]{graphics/pdf/Rule3.pdf}
        \subcaption{}
        \label{fig:earB}
      \end{subfigure}
      \caption{Illustration of \Cref{rr:ears} \textbf{(a)} and \Cref{rr:ears-2} \textbf{(b)}.}
      \label{fig:ears}
    \end{figure}

\newcommand{\ori}{\ensuremath{O_{\circlearrowright}}\xspace}
\newcommand{\ole}{\ensuremath{O_{\circlearrowleft}}\xspace}
\begin{restatable}\restaterefName{rr:ears}{Solution Normalization 2}{lemma}{rrEars}
    \label{rr:ears}
    Let $X$ be a subinstance of $\mathcal I$ that has a tidy solution~$(\mathcal E, \mathcal A)$.
    Let $O \subseteq X \setminus C$ be a set of at least $4|C|+10$ equivalent upper $xy$-ears (lower $xy$-ears) that belong to a different cluster than $x$ and $y$. 
    Then there exists a subinstance $X' \subsetneq X$ of $\mathcal I$ that has a tidy solution.
\end{restatable}
\begin{prooflater}{proofEars}
    Let $\mu$ denote the cluster containing the vertices of $O$.
    For each $v \in O$, consider the directed cycle consisting of the directed edges $xv$ and $vy$, together with an imaginary monotone edge~$yx$.
    If this cycle is oriented in clockwise direction in $(\mathcal E, \mathcal A)$ we define $v$ to be a \emph{clockwise ear} (see $a$, $b$, and $c$ in \Cref{fig:ears} for an example), otherwise it is a \emph{counterclockwise ear}.
    Let \ori and \ole denote the sets of clockwise and counterclockwise ears of $O$, respectively.

    Since $|O| \geq 4|C|+10$, we have $\max\{|\ori|, |\ole|\} \geq 2|C|+5$.
    Without loss of generality, we assume $|\ori| \geq 2|C|+5$.
    Note that $\ori$ does not contain two vertices on the same level and thus $\ell$ induces a linearly ordered sequence $\phi$ whose elements are $\ori$.
    For two vertices $a, b$ that are consecutive in $\phi$, the corresponding cycle $xayb$ contains no other vertices of $O$ in its interior.
    We say that the cycle $xayb$ is \emph{good} if its interior contains an augmentation path between $a$ and $b$ and no vertices of $C$.
    Since $T^{-1}(\mu)$ is connected in $(\mathcal E, \mathcal A)$ and $\mu \notin \{T(x), T(y)\}$, there is at most one consecutive pair $a, b$ of vertices in $\phi$ such that the interior of the cycle $xayb$ contains no augmentation path between $a$ and $b$, hence the consecutive pairs of $\phi$ induce at most $|C| + 1$ cycles that are not good (at most one without augmentation path and at most $|C|$ with vertices of $C$).
    Since $|O| \geq 2|C|+5 = 2(|C| + 1) + 3$, by \Cref{prop:sequences} there are consecutive vertices $a, b, c$ in $\phi$ such that the cycles $axby$ and $bxcy$ are both good.

Let $F$ denote the region bounded by the cycle $axcy$.
    Since $axby$ and $bxcy$ are good cycles,~$F$ contains no vertices of $C$ and an augmentation path between $a$ and $c$.
    Moreover, recall that $\ell(a) < \ell(b) < \ell(c)$ by the definition of $\phi$.
    We can thus remove $b$ and all other vertices contained in $F$ from $(\mathcal E, \mathcal A)$ and add a direct augmentation edge $ac$ inside $F$ to obtain an \yclp solution $(\mathcal E', \mathcal A')$ of a smaller subinstance~$X'$; see \Cref{fig:earA}.
    It remains to show that this solution is tidy.

    First observe that $a$, $b$, and $c$ are all upper or all lower ears and that the two faces incident to $ac$ are triangles.
    Since $a, b$, and $c$ are equivalent, it is clear that $b$ is covered by $ac$.
    Every other ear $v \neq b \in X \setminus C$ is covered by the same augmentation edge $ut$ in $\mathcal A$ and $\mathcal A'$, unless $ut \in \{ab, bc\}$.
    We consider the case $ut = ab$, the other case is analogous.
    Since $v$ is covered by $ab$ in $\mathcal A$, it is $\ell(a) < \ell(v) < \ell(b) < \ell(c)$ and thus $v$ is covered by $ac$ in $\mathcal A'$.
    Since every pair of ears that lie on the same level and have the same neighborhood is covered by two augmentation edges in $\mathcal A$, the same also holds for $\mathcal A'$ and therefore $(\mathcal E', \mathcal A')$ is tidy.
\end{prooflater}

\begin{restatable}\restaterefName{rr:ears-2}{Solution Normalization 3}{lemma}{rrEarsTwo}
    \label{rr:ears-2}
    Let $X$ be a subinstance of $\mathcal I$ with a tidy solution~$(\mathcal E, \mathcal A)$.
    If $X \setminus C$ contains at least $4|C| + 5(4|C|+10) \cdot |T| + 6$ upper $xy$-ears (lower $xy$-ears), then there exists a subinstance $X' \subsetneq X$ of $\mathcal I$ that has a tidy solution.
\end{restatable}
\begin{prooflater}{proofEarsTwo}
    Without loss of generality we consider upper $xy$-ears, the other case is symmetric.
    Let~$K$ be the set of all upper $xy$-ears and let $O \subseteq K$ be the set of upper $xy$-ears that lie in the same cluster as $x$ or $y$.
    Because there are at most $|T|$ equivalence classes that contain $xy$-ears, we have $|K \setminus O| \leq (4|C|+10) \cdot |T|$ after applying \Cref{rr:ears} exhaustively.
    Hence $|K| \geq 4|C| + 5(4|C|+10) \cdot |T| + 6 \geq  4|C| + 5|K \setminus O| + 6$ and therefore $|O| = |K| - |K \setminus O| \geq 4|C| + 4|K \setminus O| + 6 = 4(|C| + |K \setminus O|) + 6$.

    Let \ori and \ole denote the clockwise and counterclockwise upper ears of $O$ as in the proof of \Cref{rr:ears-2}.
    Note that we have $\max\{|\ori|, |\ole|\} \geq 2 (|C| + |K \setminus O|) + 3$.
    Without loss of generality, we assume $|\ori| \geq 2 (|C| + |K \setminus O|) + 3$.
    Let $\phi$ denote the ordered sequence of the elements of $\ori$ induced by $\ell$.
    In contrast to the proof of \Cref{rr:ears-2}, we cannot argue that the cycles induced by consecutive ears of $O$ contain augmentation paths, since every ear of $O$ now lies in the same cluster as $x$ or $y$.
    However, since $|\phi| \geq 2 (|C| + |K \setminus O|) + 3$, by \Cref{prop:sequences} there are three consecutive ears $a, b, c$ in $\phi$ such that the cycles $axby$ and $bxcy$ contain no vertex of~$C$ and no ear of $K \setminus O$ in their interior.
    As above, we can thus remove~$b$ and everything contained in the interior of the cycle $axcy$ from $(\mathcal E, \mathcal A)$ and add a direct augmentation edge $ac$; see \Cref{fig:earB}.
    Observe that all clusters remain connected, since $a$ and~$c$ both belong to the same cluster as $x$ or $y$.
    In particular, if $T(x) = T(y)$, then $T(x) = T(y) = T(a) = T(c)$ and thus $x$ and $y$ remain connected through $a$ and $c$
    Hence we obtain an \yclp solution~$(\mathcal E', \mathcal A')$ of a smaller subinstance~$X'$.
    Note that, while the augmentation edge $ac$ is superfluous for connecting the vertices of the cluster, it ensures that~$(\mathcal E', \mathcal A')$ is tidy as we show in the following.

    Note that the two faces incident to $ac$ are triangles.
    Moreover, since every ear of~$O$ lies the same cluster as $x$ or $y$, we have $\{T(a), T(b), T(c)\} \subseteq \{T(x), T(y)\}$ and thus $ac$ covers $b$.
    As above, every other ear that is covered by an augmentation edge distinct from $ab$ and $bc$ in~$(\mathcal E, \mathcal A)$ remains covered by the same augmentation edge in~$(\mathcal E', \mathcal A')$.
    Consider the case where an ear $v \neq b \in X \setminus C$ is covered by $ab$, the case where $v$ is covered by $bc$ is symmetric.
    Since $ab$ covers $v$, it is $\ell(a) < \ell(v) < \ell(b) < \ell(c)$ and $a$, $v$, and $b$ are equivalent or $\{T(a), T(v), T(b)\} \subseteq \{T(x), T(y)\}$.
    In both cases, we obtain $\{T(a), T(v), T(c)\} \subseteq \{T(x), T(y)\}$ and thus $ac$ covers~$v$.
    Therefore, the solution $(\mathcal E', \mathcal A')$ is tidy.
\end{prooflater}

For each $\{x,y\} \in \binom{C}{2}$, \Cref{rr:ears-2} bounds the number of upper (lower) $xy$-ears and for each cluster $\mu$ of $T$, \Cref{rr:transversal} additionally bounds the number of $xy$-transversals that belong to $\mu$.
Hence, we obtain the following result.

\begin{corollary}
    \label{cor:no-deg2}
    Let $X$ be a subinstance of $\mathcal I$ with a tidy solution that has been exhaustively reduced by \Cref{rr:transversal,rr:ears,rr:ears-2}.
    Then the number of degree-2 vertices in $X$ is in $O(|C|^3 \cdot |T|)$.
\end{corollary}

It remains to bound the number of vertices of degree at most 1.
With the following solution normalization, we first reduce such vertices of aug-degree at most 1.

\begin{lemma}[Solution Normalization 4]
    \label{rr:deg1}
    Let $X$ be a subinstance of $\mathcal I$ that has a tidy solution~$(\mathcal E, \mathcal A)$.
    If there exists a vertex $v \in X\setminus C$ of degree at most 1 such that either
    \begin{itemize}
        \item $v$ is an isolated vertex of aug-degree at most 1, or
        \item $v$ is a $v'$-pendant vertex and (i) $v$ has aug-degree 0 or (ii) $v$ has aug-degree 1 and $T(v) \neq T(v')$,
    \end{itemize}
    then $X \setminus \{v\}$ also admits a tidy solution.
\end{lemma}
\begin{proof}Note that, in each of the cases, $v$ has at most one neighbor in $\mathcal E \cup \mathcal A$ that belongs to the same cluster.
    Hence after removing $v$ from $(\mathcal E, \mathcal A)$, all clusters remain connected.
    Since~$v$ is not an ear and since its augmentation edge cannot cover any ear (as it does not connect two ears), the solution remains tidy.
\end{proof}

\newcommand{\phil}{\ensuremath{\phi_{\text{l}}}\xspace}
\newcommand{\phir}{\ensuremath{\phi_{\text{r}}}\xspace}

While degree-2 vertices are fairly simple to reduce since many equivalent degree-2 vertices also form many cycles with their two common neighbors, reducing degree-1 vertices is much more difficult. 
For a solution $(\mathcal E, \mathcal A)$, an \emph{induced augmentation path} is an induced path in $\mathcal A$ whose vertices all belong to $X \setminus C$, have aug-degree 2, and degree at most 1.
In particular, note that the endpoints of such an induced augmentation path $\pi$ also have aug-degree 2 and thus have an augmentation edge to a vertex outside of $\pi$; see \Cref{fig:shrinkingA} for an example.
After exhaustively applying the normalizations above, we can show that, if many degree-1 vertices remain, we also find an induced augmentation path that contains many degree-1 vertices.
For such a path, we can again find a sufficient number of cycles along which we can reduce; see \Cref{fig:deg-1-paths}.

\begin{figure}
    \begin{subfigure}{.32\textwidth}
        \centering
        \includegraphics[page=1,scale=0.9]{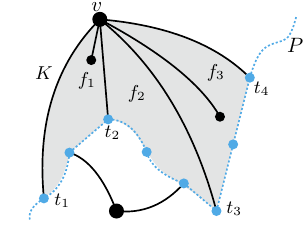}
        \subcaption{}
        \label{fig:deg-1-pathsA}
    \end{subfigure}
    \begin{subfigure}{.32\textwidth}
        \centering
        \includegraphics[page=2,scale=0.9]{graphics/pdf/Rule5}
        \subcaption{}
        \label{fig:deg-1-pathsB}
    \end{subfigure}
    \begin{subfigure}{.32\textwidth}
        \centering
        \includegraphics[page=3,scale=0.9]{graphics/pdf/Rule5}
        \subcaption{}
        \label{fig:deg-1-pathsC}
    \end{subfigure}
    \caption{Illustration of \Cref{rr:deg-1-paths}. 
    \textbf{(a)} The region $K$ bounded by the three consecutive cycles $f_1, f_2, f_3$.
    \textbf{(b)} A corresponding straight-line drawing where $K$ is a simple polygon. One possible diagonal $ab$ that can be used to shortcut $P$ is highlighted as a dashed red edge.
    \textbf{(c)} The drawing obtained after removing all vertices inside $K$ and shortcutting $P$ along $ab$.
    }
    \label{fig:deg-1-paths}
\end{figure}

\begin{restatable}\restaterefName{rr:deg-1-paths}{Solution Normalization 5}{lemma}{rrDegOnePaths}
    \label{rr:deg-1-paths}
    Let $X$ be a subinstance of $\mathcal I$ with a tidy solution~$(\mathcal E, \mathcal A)$.
    If there is an induced augmentation path $P$ that contains at least $8(|C|^2+|C| + 1)$ pendant vertices, there exists a subinstance $X' \subsetneq X$ of $\mathcal I$ that has a tidy solution.
\end{restatable}
\begin{proofsketch}
    Consider the situation shown in \Cref{fig:deg-1-pathsA}, where a vertex $v \in C$ together with a subpath of $P$ forms three consecutive cycles $f_1$, $f_2$, and $f_3$ that contain no vertex of $C$.
    By the work of Pach and T{\'{o}}th~\cite{pach_monotone_2004}, we can assume that the solution $(\mathcal E, \mathcal A)$ is a straight-line drawing and thus the region $K$ that is the union of the interiors of $f_1$, $f_2$, and $f_3$ is a simple polygon of size at least~5; see \Cref{fig:deg-1-pathsB}.
    By \Cref{prop:polygon-diagonal}, we thus find a diagonal of $K$ that is not incident to $v$ along which we can shortcut to obtain a tidy solution of a smaller subinstance; see \Cref{fig:deg-1-pathsC}.
    If $P$ contains sufficiently many pendant vertices, we can show that we always find three such consecutive cycles due to the planarity of $(\mathcal E, \mathcal A)$.
\end{proofsketch}
\begin{prooflater}{proofDegOnePaths}
    An augmentation edge $e$ is \emph{superfluous} in a tidy solution $(\mathcal E, \mathcal A)$ if $(\mathcal E, \mathcal A - e)$ is also a tidy solution.  If $(\mathcal E, \mathcal A)$ contains superfluous augmentation edges, we first remove them.
    We can thus assume that all vertices of $P$ belong to the same cluster, since no ears are involved.
    As the next step, for each vertex $v \in C$, pick an arbitrary pendant neighbor of $v$ in $P$ (if one exists) and mark it as \emph{critical}.
    We will ensure that these critical vertices are also contained in the subinstance $X'$ and use them to ensure that the clusters remain connected.
    Note that there are at most $|C|$ critical vertices.

    Since the first and the last vertex of $P$ are each incident to exactly one augmentation edge that does not belong to~$P$, these two augmentation edges partition the edges of $\mathcal E$ incident to~$P$ into edges that lie on the left and the right side of~$P$, respectively.
    Let \phil be the sequence of vertices obtained by enumerating the endpoints of all edges incident to the left side of~$P$ in the order induced by $\mathcal E$.
    Analogously, we define the sequence \phir for the right side of~$P$.
    Note that all elements of $\phil$ and $\phir$ are contained in $C$.
    Moreover, since every pendant of~$P$ contributes exactly one element to $\phil$ or $\phir$, we have $\max\{|\phil|, |\phir|\} \geq 4(|C|^2 + |C| + 1)$.
    Without loss of generality, we assume $|\phil| \geq 4(|C|^2 + |C| + 1)$.
    
    Assume that there exist two distinct vertices $a, b \in C$ such that $\phil$ contains the consecutive sequence $ab$ more than once.
    Let $e_a^1, e_b^1, e_a^2, e_b^2$ denote the corresponding edges in the order induced by $\phil$.
    Consider the cycle $c$ induced by $e_a^1$, $e_a^2$, and the subpath of $P$ connecting them.
    Note that the given order implies that $e_b^1$ and $e_b^2$ lie on different sides of $c$ in $(\mathcal E, \mathcal A)$.
    But since $e_b^1$ and $e_b^2$ share the endpoint $b$ that does not lie on $c$, this is a contradiction to the planarity of $(\mathcal E, \mathcal A)$.
    Therefore, $\phil$ contains each of the $|C| \cdot (|C| -1)$ distinct consecutive pairs at most once.
    Since $|\phil| \geq 4(|C|^2 + |C| + 1)$, note that we find $|C|^2 + |C| + 1$ disjoint consecutive quadruplets in $\phil$.
    Since $\phil$ contains each distinct consecutive pair at most once, there are at least $|C|^2 + |C| + 1 - |C| \cdot (|C| -1) = 2|C| + 1$ such quadruplets that contain four copies of the same vertex.
    We refer to these quadruplets as \emph{uniform quadruplets}.
    
    We now consider uniform quadruplets in $(\mathcal E, \mathcal A)$.
    Since $G$ has no multi-edges, each uniform quadruplet corresponds to four distinct pendant vertices $t_1, t_2, t_3, t_4$ in $P$ that are adjacent to a common vertex $v \in C$.
    Note that $t_i$ and $t_{i+1}$ need not be adjacent in $P$, since there may be vertices between them that are isolated or whose corresponding edges lie to the right of~$P$.
    However, the edges $vt_1$, $vt_2$, $vt_3$, $vt_4$, together with the subpath of $P$ between $t_1$ and $t_4$ form three consecutive cycles $f_1$, $f_2$, and $f_3$ of $(\mathcal E, \mathcal A)$; see \Cref{fig:deg-1-pathsA}.

Note that the interiors of cycles corresponding to disjoint uniform quadruplets cannot be nested, as this would contradict the assumption that the edges $vt_1$, $vt_2$, $vt_3$, and $vt_4$ all lie to the left of~$P$.
    Therefore, because there are at least $2|C|+1$ uniform quadruplets, there exists a uniform quadruplet such that $f_1$, $f_2$, and $f_3$ enclose no vertex of $C$.
    Moreover, observe that every vertex of $V(G) \setminus C$ enclosed by one of these cycles can only be connected to $v$.
    In particular, $f_1$, $f_2$, and $f_3$ enclose no ears and no augmentation edges that cover ears of $V(G) \setminus X$.
    Therefore, we can remove everything enclosed by $f_1$, $f_2$, and $f_3$ and we obtain a tidy solution where $f_1$, $f_2$ and $f_3$ are empty.

    Consider the region $K$ of $(\mathcal E, \mathcal A)$ that is the union of the interiors of $f_1$, $f_2$, and $f_3$, i.e., the boundary of $K$ consists of $t_4$, $v$, $t_1$, and the subpath of $P$ between $t_1$ and $t_4$.
    By the work of Pach and T{\'{o}}th~\cite{pach_monotone_2004}, we can assume that all edges of $(\mathcal E, \mathcal A)$ are drawn as straight line segments without changing the $y$-coordinates of the vertices and thus $K$ can be considered as a polygon of size at least 5; see \Cref{fig:deg-1-pathsC}.
    Note that the interior of $K$ only contains the edges $vt_2$ and $vt_3$ as straight line segments and is otherwise empty.
    By \Cref{prop:polygon-diagonal}, there are vertices $a, b \neq v$ on the boundary of $K$ that are non-adjacent and can be connected by a $y$-monotone edge in~$K$.
    Since $a, b \neq v$, $a$ and $b$ are contained in $P$.
    Remove all vertices that lie between $a$ and $b$ in $P$ from the solution and ``shortcut'' $P$ by adding the $y$-monotone augmentation edge $ab$; see \Cref{fig:deg-1-pathsC}.
    If neither $t_2$ nor $t_3$ lie between $a$ and $b$ in $P$, then $a$ and $b$ share one of the cycles $f_1$, $f_2$, or $f_3$ and thus $ab$ can be drawn entirely within the corresponding face without introducing crossings with $vt_2$ or $vt_3$.
    If $t_2$ or $t_3$ lies between $a$ and $b$ in $P$, then the vertex was removed along with its incident edge and we thus also obtain no crossings in this case.
    Because $a$ and $b$ are not adjacent, we removed at least one vertex between $a$ and $b$ in $P$ and we thus obtain a solution for a smaller subinstance $X'$.
    In particular, since we did not remove any critical vertices of $P$, each vertex that had a neighbor in $P$ before still has a neighbor in the shortened augmentation path and thus the cluster of $P$ remains connected.
    Since no ears are involved, also observe that the solution remains tidy.
\end{prooflater}

\begin{figure}
    \begin{subfigure}{.5\textwidth}
        \centering
        \includegraphics[page=1]{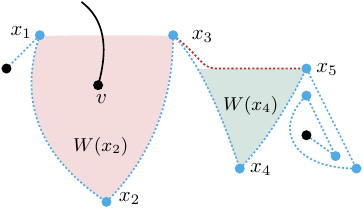}
        \subcaption{}
        \label{fig:shrinkingA}
    \end{subfigure}
    \begin{subfigure}{.5\textwidth}
        \centering
        \includegraphics[page=2]{graphics/pdf/shrinkingProp.pdf}
        \subcaption{}
        \label{fig:shrinkingB}
    \end{subfigure}
    \caption{\textbf{(a)} 
    An induced augmentation path $\pi$ (blue vertices) that consists of isolated vertices and has the shrinking property.
The wedges $W(x_2)$ and $W(x_4)$ of the valleys $x_2$ and $x_4$ are highlighted in red and green, respectively.
    The valley $x_2$ is blocked by $v$.
    Because $x_4$ is not blocked by any vertex, the augmentation edge $x_3x_5$ (red dashed edge) can be inserted crossing-free in a $y$-monotone fashion.
    Thus $x_4$ can be removed to obtain a solution of a smaller subinstance \textbf{(b)}.
    }
    \label{fig:shrinking}
\end{figure}

It remains to bound isolated vertices.
This is the most difficult step, as isolated vertices are not connected to the core and thus the approaches in the normalizations above do not work.
However, we show that, if a solution contains many isolated vertices after applying the normalizations above, these isolated vertices decompose into few induced augmentation paths that have a special structure, which we can use to bound the length of these paths.

Let~$\pi$ be an induced augmentation path.  A vertex~$v$ on~$\pi$ is a \emph{peak} if both augmentation-neighbors of $v$ lie on a level below~$\ell(v)$ and it is a \emph{valley} if both its augmentation-neighbors lie on a level above~$\ell(v)$.  We say that~$\pi$ has the \emph{shrinking property} if (i) its vertices are alternatingly peaks and valleys and (ii) it can be oriented such that for each peak~$p$ of~$\pi$ all peaks~$p'$ later on the path satisfy~$\ell(p') \le \ell(p)$ and for each valley~$v$ of~$\pi$ all valleys~$v'$ later on the path satisfy~$\ell(v') \ge \ell(v)$; see \Cref{fig:shrinkingA} for an example.

\begin{figure}
    \begin{subfigure}{.37\textwidth}
        \centering
        \includegraphics[page=1]{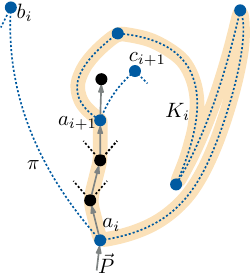}
        \subcaption{}
        \label{fig:non-separatingA}
    \end{subfigure}
    \begin{subfigure}{.37\textwidth}
        \centering
        \includegraphics[page=2]{graphics/pdf/non-separating.pdf}
        \subcaption{}
        \label{fig:non-separatingB}
    \end{subfigure}
    \begin{subfigure}{.24\textwidth}
        \centering
        \includegraphics[page=3]{graphics/pdf/non-separating.pdf}
        \subcaption{}
        \label{fig:non-separatingC}
    \end{subfigure}
    \caption{ Examples where the cycle $K_i$ (highlighted in orange) is separating \textbf{(a)} and non-separating~\textbf{(b)}. 
    The paths $\pi$ and $\vec P$ are represented by blue dashed edges and gray directed edges, respectively.
    The region $R_i$ corresponding to the non-separating cycle $K_i$ is highlighted in light green.
    \textbf{(c)} If $K_i$ is non-separating and no path of $\mathcal P$ ends in the region $R_i$, then there exist two adjacent vertices of $\vec P[a_i, a_{i+1}]$ that belong to the same path of $\mathcal P$.
    }
    \label{fig:non-separating}
\end{figure}

\begin{figure}
    \begin{subfigure}{.34\textwidth}
        \centering
        \includegraphics[page=1]{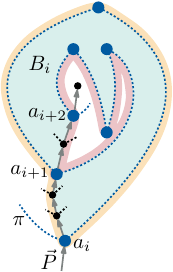}
        \subcaption{}
        \label{fig:separatingA}
    \end{subfigure}
    \begin{subfigure}{.32\textwidth}
        \centering
        \includegraphics[page=2]{graphics/pdf/separating.pdf}
        \subcaption{}
        \label{fig:separatingB}
    \end{subfigure}
    \begin{subfigure}{.32\textwidth}
        \centering
        \includegraphics[page=3]{graphics/pdf/separating.pdf}
        \subcaption{}
        \label{fig:separatingC}
    \end{subfigure}
    \caption{\textbf{(a)} Cycles $K_i$ (yellow) and $K_{i+1}$ (red) that are both separating and the corresponding region $B_i$. \textbf{(b)} If no path in $\mathcal P$ ends in the region $B_i$, then each path that enters $B_i$ through $\vec P[a_i,a_{i+1}]$ leaves it through $\vec P[a_{i+1},a_{i+2}]$ and vice versa. A tidy solution of a smaller subinstance can be obtained by replacing $\pi[a_i,a_{i+1}]$ and all paths in $B_i$ with single $y$-monotone edges in the region that is the union of $S_0$, $S_1$, and $S_2$ \textbf{(c)}. 
    }
    \label{fig:separating}
\end{figure}

\begin{restatable}\restaterefName{rr:isolated}{Solution Normalization 6}{lemma}{rrIsolated}
\label{rr:isolated}
    Let $X$ be a subinstance of $\mathcal I$ with a tidy solution~$(\mathcal E, \mathcal A)$, let 
    $\mathcal P$ be a set of pairwise disjoint and non-adjacent induced augmentation paths that satisfy the shrinking property and consist of isolated vertices.  Let ${L = \bigcup_{P \in \mathcal P} V(P) \subseteq X}$ be the set of vertices contained in a path of $\mathcal P$.
    If $|L| \geq {6(|X \setminus L|^2 +1) \cdot |\mathcal P|}$, there exists a subinstance $X' \subsetneq X$ of $\mathcal I$ that has a tidy solution.
\end{restatable}
\begin{proofsketch}
    Since every path of $\mathcal P$ has the shrinking property, every vertex $v \in L$ is either a valley or a peak.
    We consider the region $W(v)$ bounded by the two incident edges of $v$ and a horizontal line on the level of the neighbor whose level is closer to the level of $v$ (see \Cref{fig:shrinkingA}) and we say that a vertex $w$ \emph{blocks} $v$ if $w$ lies in $W(v)$.
    If a vertex $v \in L$ is not blocked by any vertex, we find a smaller subinstance that admits a tidy solution as illustrated in \Cref{fig:shrinkingB}, hence we assume that every vertex of $L$ is blocked.

    We define a directed \emph{blocker graph $\mathcal B$} on $V(X)$ by choosing for each $v \in L$ a vertex $u \in X$ minimizing~$|\ell(v) - \ell(u)|$ that blocks $v$ and inserting the directed edge $vu$.
    We can show that~$\mathcal B$ is a DAG whose sinks are the vertices of $X \setminus L$ and that consists of a bounded number of edge-disjoint paths.
    Since $|L|$ is sufficiently large, $\mathcal B$ contains a long directed path $\vec P$ and there exists a path $\pi \in \mathcal P$ such that $\vec P$ contains many vertices of $\pi$. We denote these vertices by~$a_1,\dots, a_t$ in the order in which they occur along~$\vec P$; see \Cref{fig:non-separating}.
    For $i \in \{1,\dots,t-1\}$, we denote by~$K_i$ the cycle formed by~$\pi[a_i,a_{i+1}]$ and~$\vec P[a_i,a_{i+1}]$.  Let~$b_i$ denote the neighbor of~$a_i$ on~$\pi$ that is not in~$K_i$ and let~$c_{i+1}$ be the neighbor of~$a_{i+1}$ on~$\pi$ that is not in~$K_i$.  We say that~$K_i$ is \emph{separating} if~$b_i$ and~$c_{i+1}$ lie on different sides of~$K_i$ in~$(\mathcal E',\mathcal A')$ and \emph{non-separating} otherwise; see \Cref{fig:non-separating} for examples.

    If~$K_i$ is non-separating, we let~$R_i$ denote the region enclosed by~$K_i$ that does not contain~$b_i$ and~$c_{i+1}$; see \Cref{fig:non-separatingB}.
    If $R_i$ contains no vertex of $X \setminus L$, no path of $\mathcal P$ can end in the region $R_i$ and thus each path that enters the region through a vertex of $\vec P[a_i,a_{i+1}]$ also leaves it through a vertex of $\vec P[a_i,a_{i+1}]$; see \Cref{fig:non-separatingC}.
    Due to planarity, this means that we find two vertices that belong to the same path of $\pi' \in \mathcal P$ and are adjacent in $\vec P$.
    Due to the definition of an edge in $\vec P$, we can insert a $y$-monotone augmentation edge between these two vertices and thus shortcut $\pi'$ for a tidy solution of a smaller subinstance.

    It thus remains to deal with separating $K_i$.
    Consider $i \in \{1, \dots, t-2\}$ such that both $K_i$ and $K_{i+1}$ are separating.
    Then we define the region $B_i$ as the connected region bounded by $K_i$ and $K_{i+1}$; see \Cref{fig:separatingA}.
    If $B_i$ contains no vertex of $X \setminus L$, no path of $\pi$ ends in $B_i$ as in the non-separating case.
    If a path of $\pi$ enters and leaves $B_i$ through vertices of $\vec P[a_i,a_{i+1}]$ (of $\vec P[a_{i+1},a_{i+2}]$), we again find two adjacent vertices of $\vec P$ that belong to the same path of $\mathcal P$ and can shortcut as in the non-separating case.
    Hence every path enters $B_i$ through $\vec P[a_i,a_{i+1}]$ and leaves it through $\vec P[a_{i+1},a_{i+2}]$; see \Cref{fig:separatingB}.
    We can then remove all vertices that lie between $a_i$ and $a_{i+1}$ and all paths in $B_i$ and replace each of them with a single augmentation edge that we route through the region $S_0 \cup S_1 \cup S_2$ highlighted in \Cref{fig:separatingB} to obtain a tidy solution to a smaller subinstance; see \Cref{fig:separatingC}.

    Because $|L|$ (and thus $t$) is sufficiently large compared to $|X \setminus L|$, we can show that there exists an $i \in \{1, \dots, t-2\}$ such that (i) $K_i$ is non-separating and $R_i$ contains no vertex of $X \setminus L$, or (ii) $K_i$ and $K_{i+1}$ are separating and $B_i$ contains no vertex of $X \setminus L$.
    As shown above, we find a smaller subinstance that admits a tidy solution in both cases.
\end{proofsketch}
\begin{prooflater}{proofIsolated}
    As above, after removing superfluous augmentation edges, we may assume that the vertices of each path of $\mathcal P$ belong to a single cluster.
    Since all vertices of paths in $\mathcal P$ have aug-degree~2 and the paths in $\mathcal P$ are pairwise disjoint and non-adjacent, each path of $\mathcal P$ has two augmentation edges connecting it to vertices of $X \setminus L$.
    We call these vertices the \emph{attachments} of the path.
Consider a vertex~$v$ of a path in~$\mathcal P$.
    We refer to the neighbor of $v$ whose level is closer to the level of $v$ as $\cl(v)$, the other neighbor is called $\fr(v)$; we break ties arbitrarily.
    The \emph{wedge $W(v)$} of $v$ is the region enclosed by the two augmentation edges incident to $v$ and a horizontal line~$l$ through $\cl(v)$ in the solution~$(\mathcal E,\mathcal A)$; see \Cref{fig:shrinkingA}.  We define the wedge~$W(v)$ so that the parts of the augmentation edges that bound~$W(v)$ are not part of the region, but the part of~$l$ that bounds the region is. A vertex~$w$ \emph{blocks} the vertex~$v$ if~$w$ lies in $W(v)$ in~$(\mathcal E,\mathcal A)$.  
    
    If vertex~$v$ is not blocked by any vertex, then~$W(v)$ is empty in~$(\mathcal E,\mathcal A)$. In this case, we can insert a direct $y$-monotone edge between $\fr(v)$ and $\cl(v)$ following the edge $\fr(v)v$ and the boundary of $W(v)$; see \Cref{fig:shrinkingB}. We can then remove $v$ and obtain a tidy solution of a smaller subinstance.  In the following, we thus assume that every vertex of $L$ is blocked.

    We define a directed \emph{blocker graph $\mathcal B$} on $V(X)$ by choosing for each $v \in L$ a vertex $u \in X$ minimizing~$|\ell(v) - \ell(u)|$ that blocks $v$ and inserting the directed edge $vu$.
    
    \begin{claim}
    \label{cl:blocker-properties}
        The blocker graph $\mathcal B$ has the following properties:
        \begin{enumerate}
            \item every sink is contained in $X \setminus L$,
            \item if a valley $v$ has an outgoing edge $vu$, then $\ell(v) < \ell(u)$ and $u$ is not a peak,
            \item if a peak $v$ has an outgoing edge $vu$, then $\ell(v) > \ell(u)$ and $u$ is not a valley,
            \item $\mathcal B$ is acyclic, 
            \item every vertex has outdegree at most 1,
            \item every peak and every valley has indegree at most 1 and every vertex of $X \setminus L$ has indegree at most 2,
            \item\label{it:common} for each directed path $\vec P$ of $\mathcal B$ and each path $\pi \in \mathcal P$, $\vec P$ and $\mathcal B$ visit their common vertices in the same order, and
            \item for each directed path $\vec P$ of $\mathcal B$, the edges of $\vec P$ can be inserted planarly and in a $y$-monotone fashion into $(\mathcal E, \mathcal B)$.
        \end{enumerate}
    \end{claim}
    \begin{claimproof}[Proof of claim]
        Since we assume that every vertex of $L$ is blocked, it follows that all sinks of $\mathcal B$ are in $X \setminus L$.  This shows the first property.

        For the second property, let~$vu$ be an edge such that~$v$ is a valley.  Since~$u$ blocks $v$, we have~$\ell(v) < \ell(u)$.  Assume for the sake of contradiction that $u$ is a peak. Since $u$ blocks $v$, $u$ is the vertex in the region $W(v)$ with the lowest level. But since $u$ is a peak, the two neighbors of $u$ lie on a strictly lower level than $u$.
        Due to planarity and $y$-monotonicity, both these neighbors lie inside $W(v)$ as well.  This contradicts the fact that $u$ is the lowest vertex in this region.  This concludes the proof of the second property.  The third property follows by an analogous argument.
        
        For the fourth property, observe that, by the second and third properties, the levels of vertices on any directed path of $\mathcal B$ strictly increase or decrease and thus $\mathcal B$ contains no directed cycles.
        
        The fifth property follows immediately from the construction of~$\mathcal B$.

        For the sixth property, consider a vertex~$v$ and assume that it has incoming $\mathcal B$-edges~$uv$ and~$wv$ from two valleys~$u,w$.  By the second property, we have~$\ell(u), \ell(w) < \ell(v)$.  We thus have~$u \notin W(w)$ and~$w \notin W(u)$ as otherwise, we would have a contradiction to the fact that~$v$ is a lowest vertex in these regions. 
        Because $v \in W(w) \cap W(u)$, $u \notin W(w)$, $w \notin W(u)$ and $\ell(u), \ell(w) < \ell(v)$, the points in which the two $y$-monotone augmentation edges incident to $w$ cross level $\ell(v)$ alternate with the points in which the two $y$-monotone augmentation edges incident to~$u$ cross level $\ell(v)$.
        But since $\ell(u), \ell(w) < \ell(v)$, this means that two of these augmentation edges cross, a contradiction.
Thus, any vertex has at most one incoming~$\mathcal B$-edge from a valley and, by symmetry, at most one incoming $\mathcal B$-edge from a peak. Since peaks have incoming $\mathcal B$-edges only from peaks and valleys have incoming $\mathcal B$-edges from valleys by properties 2 and 3, property 6 follows.

        For the seventh property, consider a directed path $\vec P$ of $\mathcal B$, a path $\pi \in \mathcal B$, and let $v$ denote the first vertex of $\vec P$ that is also contained in $\pi$.
        Without loss of generality, assume that $v$ is a valley.
        By property 2, all subsequent vertices of $\vec P$ lie on a strictly higher level than $v$.
        Because $\pi$ has the shrinking property, the same also holds for all vertices after $v$ in $\pi$.
        Hence no vertex that lies before $v$ in $\vec P$ lies before $v$ in $\pi$ and thus $v$ is the first common vertex in $\vec P$ and $\pi$.
        By induction, it follows that $\vec P$ and $\pi$ visit their common vertices in the same order and thus property 7 holds.
        
        For the final property, consider a directed path $\vec P$ of $\mathcal B$.
        For an edge $uv$ of $\vec P$, recall that the region $W(u)$ contains no vertex whose level lies strictly between $\ell(u)$ and $\ell(v)$.
        The subregion of $W(u)$ restricted by a horizontal line on level $\ell(v)$ thus contains no vertices and, due to $y$-monotonicity, also no edges.
        Hence $uv$ can be drawn crossing-free in a $y$-monotone fashion into this region.
        By properties 2 and 3, the levels spanned by two distinct edges can only coincide at a common endpoint.
        Therefore, inserting all edges of $\vec P$ at the same time does not introduce crossings between the inserted edges.
    \end{claimproof}

    Assume that there is an edge~$uv$ in~$\mathcal B$ whose endpoints both belong to the same path~$\pi$ in~$\mathcal P$.
    We can then insert the edge~$uv$ planarly into~$(\mathcal E,\mathcal A)$ by \Cref{cl:blocker-properties} and remove all interior vertices of~$\pi[u,v]$ together with their incident augmentation edges to obtain a tidy solution to a subinstance $X' \subseteq X$.  Observe that~$\pi[u,v]$ has at least one internal vertex, since $u$ and~$v$ are both peaks or both valley and~$\pi$ has the shrinking property.  Thus,~$|X'| < |X|$.  In what follows, we thus assume that no two vertices from the same path in~$\mathcal P$ are connected by an edge in~$\mathcal B$.

    Now consider the graph~$\mathcal B$. By \Cref{cl:blocker-properties}, $\mathcal B$ is the edge-disjoint union of at most~$2\cdot |X \setminus L|$ maximal directed paths with sinks $X \setminus L$.    
    Since~$|L| \geq 6 \cdot (|X \setminus L|^2 +1) \cdot |\mathcal P|$, it follows that~$\mathcal B$ contains a directed path~$\vec P$ of length at least~$|\vec P| \geq (6 \cdot (|X \setminus L|^2 +1) \cdot |\mathcal P|)/ (2 \cdot |X \setminus L|) > 3 \cdot (|X \setminus L|+1) \cdot |\mathcal P|$.  We consider the case that all but the last vertex of~$\vec P$ are valleys; the case that they are all peaks can be handled analogously.     
    Since~$|\vec P| \geq 3 \cdot (|X \setminus L|+1) \cdot |\mathcal P|$, we have that there is a path~$\pi \in \mathcal P$ such that~$\vec P$ contains at least $t = 3 \cdot (|X \setminus L| + 1)$ vertices from~$\pi$.  We denote these vertices by~$a_1,\dots, a_t$ in the order in which they occur along~$\vec P$.  We now temporarily consider the solution~$(\mathcal E',\mathcal A')$ obtained from~$(\mathcal E, \mathcal A)$ by planarly inserting all edges of~$\vec P$, which is possible by \Cref{cl:blocker-properties}.  Thus~$\vec P$ is represented by a curve in~$(\mathcal E',\mathcal A')$.  For $i \in \{1,\dots,t-1\}$, we denote by~$K_i$ the cycle formed by~$\pi[a_i,a_{i+1}]$ and~$\vec P[a_i,a_{i+1}]$.  Let~$b_i$ denote the neighbor of~$a_i$ on~$\pi$ that is not in~$K_i$ and let~$c_{i+1}$ be the neighbor of~$a_{i+1}$ on~$\pi$ that is not in~$K_i$.  We say that~$K_i$ is \emph{separating} if~$b_i$ and~$c_{i+1}$ lie on different sides of~$K_i$ in~$(\mathcal E',\mathcal A')$ and \emph{non-separating} otherwise; see \Cref{fig:non-separating} for examples.

    If~$K_i$ is non-separating, we let~$R_i$ denote the region enclosed by~$K_i$ that does not contain~$b_i$ and~$c_{i+1}$; see \Cref{fig:non-separatingB}.  We have the following claim.

    \begin{claim}
        \label{cl:non-separating}
        For any two $K_i,K_j$ with~$i \ne j$ that are non-separating, we have~$R_i \cap R_j = \emptyset$ and for each non-separating~$K_i$ we have that~$R_i$ contains a vertex in~$X \setminus L$.  
    \end{claim}

    \begin{claimproof}[Proof of claim]
        For the first statement, assume without loss of generality that $i < j$.
It follows from property \ref{it:common} of \Cref{cl:blocker-properties} that $K_i$ and $K_j$ are either disjoint or $j=i+1$ and they share only~$a_{i+1}$.
        In the former case, since~$K_i$ and~$K_j$ are non-separating, we have that~$\pi[a_{i+1},a_j]$ is a path that connects the boundaries~$R_i$ and~$R_j$ and lies outside both regions.  
        It follows that~$R_i \cap R_j =\emptyset$. 
        In the latter case, vertex $c_{i+1}$ lies on the boundary of $R_j$ and $c_{i+1}$ lies outside of $R_i$ since $K_i$ is non-separating.
        Since $R_i$ and $R_j$ only share~$a_{i+1}$, it follows that~$R_i \cap R_j = \emptyset$.

        For the second part of the claim, assume that~$R_i$ contains no vertex of~$X \setminus L$.  Therefore, no path in~$\mathcal P$ has an attachment that lies in~$R_i$.
        Thus, every path in~$\mathcal P$ that enters~$R_i$ does so via an internal vertex of ~$\vec P' := \vec P[a_i,a_{i+1}]$ and it also leaves~$R_i$ via a (different) internal vertex of~$P'$. 
        Moreover, every interior vertex $u$ of $\vec P'$ is incident to a subpath of a path of~$\mathcal P$ in $R_i$, because $u$ is a valley that is blocked by its successor on $\vec P'$ and thus $u$ has an incident augmentation edge inside $R_i$.
        The subpaths of paths in~$\mathcal P$ that contain a vertex in~$R_i$ thus induce a perfect matching~$M$ on the interior vertices of~$\vec P'$.  
        Let~$e=uv \in M$ be such that its endpoint have minimum distance on~$\vec P'$.  Since the paths entering (and leaving) $R_i$ are pairwise non-crossing, the endpoints of the edges in~$M$ are pairwise non-alternating along~$\vec P'$, and thus~$u$ and~$v$ are in fact adjacent in~$\vec P'$; see \Cref{fig:non-separatingC}.  This, however, contradicts the assumption that no two vertices of the same path in~$\mathcal P$ are adjacent in~$\mathcal B$.
    \end{claimproof}

    By this claim, for each non-separating cycle~$K_i$, we have that~$R_i$ contains a vertex from~$|X \setminus L|$.  Since these regions are pairwise disjoint, it follows that there are at most~$|X \setminus L|$ non-separating cycles~$K_i$.  It remains to bound the number of separating cycles~$K_i$.

    Let~$i$ be such that both~$K_i$ and~$K_{i+1}$ are separating. 
    Let~$B_i$ denote the connected region bounded by $K_i$ and $K_{i+1}$; see \Cref{fig:separatingA}.  
We have the following claim.

    \begin{claim}
        \label{cl:separating}
        Let~$i \ne j$ be such that~$K_{i}, K_{i+1}, K_j$ and $K_{j+1}$ are separating. Then~$B_i \cap B_j = \emptyset$.  If there is an $i$, $1 \le i \le t-2$ such that~$K_i$ and~$K_{i+1}$ are separating and~$B_i$ contains no vertex in~$X \setminus L$, then there exists a subinstance~$X'$ of~$\mathcal I$ with $|X'| < |X|$ that admits a tidy solution.
    \end{claim}
    \begin{claimproof}[Proof of claim]
        For the first statement, assume without loss of generality that $i < j$.
        Recall from the proof of \Cref{cl:non-separating} that $K_i$ and $K_j$ are either disjoint or $j = i + 1$ and they share vertex $a_{i+1}$.
        Moreover, recall that $B_i$ is bounded by $K_i$ and $K_{i+1}$.
        Therefore, the boundaries of $B_i$ and $B_j$ are either disjoint (if $j > i + 2$), they share vertex $a_{i+2}$ (if $j = i + 2$) or they share cycle $K_{i+1}$ (if $j = i + 1$).
        Hence $B_i$ and $B_j$ are disjoint or one is contained in the other.
        However, since the cycle $K_{i+1}$ separates $b_{i+1}$ (which lies on the boundary of $B_i$) from $c_{i+1}$ (which either lies on the boundary of $B_j$ or on a subpath of $\pi$ that connects the boundaries of $B_i$ and $B_j$ but is disjoint from the regions themselves), containment is not possible.
        We therefore have that~$B_i \cap B_j = \emptyset$.

        For the second statement, assume that~$K_i$ and~$K_{i+1}$ are separating and~$B_i$ contains no vertex in~$X \setminus L$.  Thus, no path in~$\mathcal P$ has an attachment in~$B_i$.  Therefore, every path in~$\mathcal P$ that enters~$B_i$ enters it via an intenal vertex of~$\vec P[a_i,a_{i+1}]$ or~$\vec P[a_{i+1},a_{i+2}]$ and also leaves it again via such a vertex.  Consider the maximal subpaths of paths in~$\mathcal P$ that contain no vertex in the exterior of $B_i$.  These paths are pairwise non-crossing and each of them connects two vertices of~$\vec P[a_i,a_{i+2}]$ that are distinct from~$a_i,a_{i+1},a_{i+2}$.  Let~$Q$ be such a path and let~$x$ and~$y$ denote its endpoints.  If they both belong to~$\vec P[a_i,a_{i+1}]$, then all other paths that enter the region via a vertex in $\vec P[x,y]$ also leave it via vertex in this path and similarly as in the proof of \Cref{cl:non-separating} above, it follows that there is a path that enters and leaves the region via two vertices~$x',y'$ that are adjacent on~$\vec P$.  Since no two adjacent vertices of~$\mathcal B$ may belong to the same path in~$\mathcal P$, this is a contradiction.  Similiarly, we can argue that no subpath can enter and leave the region~$B_i$ via~$\vec P[a_{i+1},a_{i+2}]$.  Therefore, the subpaths that traverse the region~$B_i$ induce a perfect matching $M$ between the inner vertices of~$\vec P[a_i,a_{i+1}]$ and the inner vertices of~$\vec P [a_{i+1},a_{i+2}]$.  Moreover, by planarity, the matching $M$ is non-crossing in the sense that for each pair of edges $e,e' \in M$ the order of their endpoints is the same along both paths.

\newcommand{\eps}{\varepsilon}
    We now show how to reroute the paths that traverse the region~$B_i$.  To this end, let~$\eps > 0$ be sufficiently small so that the disk $D$ with radius~$\eps$ around~$a_{i+1}$ contains no other vertices in the drawing and all edges that intersect~$D$ end at~$a_{i+1}$. Let~$S_0$ be the wedge of~$D$ that is enclosed by the last edge of~$\vec P[a_i,a_{i+1}]$ and the last edge of~$\pi[a_i,a_{i+1}]$ and does not intersect any other edge in its interior.   
    Let~$S_1$ be an $\eps$-strip around $\vec P[a_i,a_{i+1}]$ in $B_i$, i.e., $S_1$ is the maximal connected subregion of $B_i$ that contains the drawing of~$\vec P[a_i,a_{i+1}]$ in~$(\mathcal E',\mathcal A')$ on its boundary and only contains points $p$ whose horizontal projection onto $\vec P[a_i,a_{i+1}]$ has distance at most $\varepsilon$ from $p$; see \Cref{fig:separatingB}.  Analogously, let~$S_2$ be the corresponding $\eps$-strip in $B_i$ for $\vec P[a_{i+1},a_{i+2}]$.  Note that the boundaries of~$S_1$ and~$S_2$ intersect precisely in~$a_{i+1}$.
    Moreover, recall that $(\mathcal E, \mathcal A)$ does not contain the edges of $\vec P$ we temporarily added in $(\mathcal E', \mathcal A')$.
    Therefore, the region $S_0 \cup S_1 \cup S_2$ in $(\mathcal E, \mathcal A)$ is homeomorphic to a disk whose boundary contains a~$y$-monotone curve~$\rho_1$ that starts at~$a_i$ and contains the vertices of~$\vec P[a_i,a_{i+1}]$ except for~$a_{i+1}$ and a $y$-monotone curve~$\rho_2$ that starts at~$a_{i+1}$ and contains all vertices of $\vec P [a_{i+1},a_{i+2}]$ except for~$a_{i+2}$.  Moreover, for any two edges~$e,e' \in M$, their endpoints appear in the same order along~$\rho_1$ and~$\rho_2$.
    Further, $S_0 \cup S_1 \cup S_2$ is horizontally convex.  Since the vertices on~$\rho_1$ lie strictly below the vertices on~$\rho_2$ and the matching~$M$ is non-crossing, we can remove all vertices in the interior of $B_i$ as well as all interior vertices of~$\pi[a_i,a_{i+1}]$ from $(\mathcal E, \mathcal A)$ and instead draw the matching~$M$ and an edge connecting~$a_i$ to~$a_{i+1}$ crossing-free in the interior of~$S_0 \cup S_1 \cup S_2$; see \Cref{fig:separatingC}. 
    We thus have found a subinstance~$X' \subseteq X$ that admits a tidy solution.  Moreover, we have~$|X'| < |X|$ since~$\pi[a_i,a_{i+1}]$ contains at least one internal vertex that is a peak due to the shrinking property.
    \end{claimproof}

    By \Cref{cl:non-separating}, there are at most $|X \setminus L|$ distinct $i \in \{1, \dots, t-1\}$ for which $K_i$ is non-separating.
    Therefore, since $t = 3\cdot (|X \setminus L|+1)$, there are at least $t - 2 - |X \setminus L
    | = 2|X \setminus L| + 1$ distinct $i \in \{1, \dots, t-2\}$ such that $K_i$ is separating, and for at most $|X \setminus L|$ of them $K_{i+1}$ is non-separating.
    Hence there are at least $2|X \setminus L| + 1 - |X \setminus L| = |X \setminus L| + 1$ distinct $i \in \{1, \dots, t-2\}$ such that $K_i$ and $K_{i + 1}$ are separating.
    By \Cref{cl:separating}, there consequently exists an $i \in \{1, \dots, t-2\}$ such that $K_i$ and $K_{i + 1}$ are separating and $B_i$ contains no vertex of $X \setminus L$.
    Since \Cref{cl:separating} then also guarantees the existence of a smaller subinstance that admits a tidy solution, this concludes the proof.
\end{prooflater}

We now use these solution normalizations to show that every yes-instance of \yclp also has a bounded-size subinstance that admits a tidy solution.

\begin{figure}
    \begin{subfigure}{.5\textwidth}
        \centering
        \includegraphics[page=1]{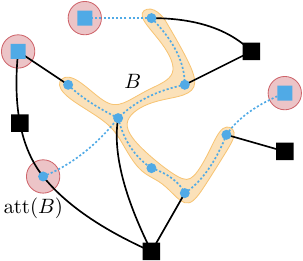}
        \subcaption{}
        \label{fig:bridgeA}
    \end{subfigure}
    \begin{subfigure}{.5\textwidth}
        \centering
        \includegraphics[page=2]{graphics/pdf/Bridge.pdf}
        \subcaption{}
        \label{fig:bridgeB}
    \end{subfigure}
    \caption{\textbf{(a)} A bridge $B$ (orange) together with its cluster attachments $\att(B)$ (red).
    Vertices of the core are represented as squares.
    \textbf{(b)} The extension~$\hat B$ of $B$.}
    \label{fig:bridge}
\end{figure}

\begin{figure}
    \begin{subfigure}{.34\textwidth}
        \centering
        \includegraphics[page=1]{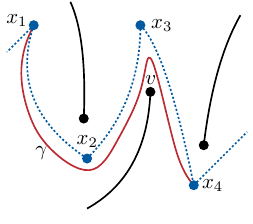}
        \subcaption{}
        \label{fig:shortcutA}
    \end{subfigure}
    \begin{subfigure}{.32\textwidth}
        \centering
        \includegraphics[page=2]{graphics/pdf/shortcut.pdf}
        \subcaption{}
        \label{fig:shortcutB}
    \end{subfigure}
    \begin{subfigure}{.32\textwidth}
        \centering
        \includegraphics[page=3]{graphics/pdf/shortcut.pdf}
        \subcaption{}
        \label{fig:shortcutC}
    \end{subfigure}
    \caption{Shortcutting an induced augmentation path $\pi$ (blue dashed edges) that consists of isolated vertices and does not satisfy the shrinking property. \textbf{(a)} Vertex $x_4$ lies on a lower level than the successor $x_2$ of a highest vertex $x_1$. Hence a curve $\gamma$ (red) can be drawn between $x_1$ and $x_4$ such that $\gamma$ does not intersect $\ell(x_1)$ and $\ell(x_4)$.
    \textbf{(b)} A corresponding solution where the curve $\gamma$ is $y$-monotone. 
    Path $\pi$ can thus be shortcut by replacing $\gamma$ with an augmentation edge $x_1x_4$ \textbf{(c)}.
    }
    \label{fig:shortcut}
\end{figure}
\begin{restatable}\restateref{lm:boundedSize}{lemma}{lmBoundedSize}
\label{lm:boundedSize}
  Let $\mathcal I=(G, \ell, T)$ be an instance of \yclp and let $C$ be its core vertex set.
  If $\mathcal I$ is a yes-instance, then $G$ has a subinstance $X$ with a tidy solution such that $|X| \leq f(|C|, |T|)$ for some polynomial function $f$.
\end{restatable}
\begin{proofsketch}
    First observe that $V(G)$ is a subinstance of $\mathcal I$ with a tidy solution.
    Let $X$ be the subinstance of $\mathcal I$ obtained by exhaustively applying the solution normalizations 1 and 3--6 from \Cref{rr:deg1,rr:transversal,rr:ears-2,rr:deg-1-paths,rr:isolated} and let $(\mathcal E, \mathcal A)$ be a corresponding tidy solution that minimizes the number of augmentation edges.
    Since vertices of degree at least 3 are contained in $C$ and the number of vertices of degree~2 is bounded by \Cref{cor:no-deg2}, it remains to bound the set $V_{\le 1}$ of vertices of degree at most 1.

    We decompose the graph into \emph{bridges}, which are maximal subsets of $V_{\le 1}$ connected by augmentation edges.
    Due to the minimality of the solution, every bridge $B$ is a tree and belongs to a single cluster~$\mu$.
    The \emph{cluster attachments $\att(B)$} are the vertices of $X \setminus V_{\le 1}$ that belong to $\mu$ and are adjacent to a vertex of $B$ in $\mathcal E$ or $\mathcal A$; see \Cref{fig:bridgeA}.
    The \emph{extension} $\hat B$ of $B$ is the subgraph of $\mathcal E \cup \mathcal A$ formed by the vertices $B \cup \att(B)$, all augmentation edges between vertices of $B$, and all (augmentation) edges between $B$ and $\att(B)$; see \Cref{fig:bridgeB}.
    
    Due to \Cref{rr:deg1} and the minimality of the solution, we can show that the number of bridges is bounded and that the extension of each bridge is a tree with a bounded number of leaves and nodes of aug-degree at least~3.
    What remains is a bounded number of (potentially long) induced augmentation paths.
    By \Cref{rr:deg-1-paths}, each one of them contains a bounded number of pendant vertices and thus the number of pendant vertices is also bounded.
    These pendant vertices further decompose all remaining isolated vertices into a bounded-size set~$\mathcal P'$ of induced augmentation paths that consist of isolated vertices.
    If one of these paths~$\pi$ contains two non-adjacent vertices $a, b$ such that $a$ lies on the maximum level of $\pi$ and every vertex that lies between $a$ and $b$ in $\pi$ lies on a strictly higher level than $b$, then we can use a known redrawing-argument~\cite[Lemma 5]{BrucknerR21} to shortcut $\pi$ and obtain a tidy solution of a smaller subinstance; see \Cref{fig:shortcut}.
    After splitting each path of $\mathcal P'$ at a vertex on the maximum level of the path, we obtain a bounded-size set $\mathcal P$ of disjoint, non-adjacent induced augmentation paths that have the shrinking property and consist of isolated vertices $L$.
    Since $|X \setminus L|$ and $|\mathcal P|$ are bounded,~$|L|$ is also bounded by \Cref{rr:isolated} and thus the same holds for $|X|$.\end{proofsketch}
\begin{prooflater}{proofBoundedSize}
First observe that $V(G)$ is a subinstance of $\mathcal I$ with a tidy solution.
    Let $X$ be the subinstance of $\mathcal I$ obtained by exhaustively applying the solution normalizations in \Cref{rr:deg1,rr:transversal,rr:ears,rr:ears-2,rr:deg-1-paths} and let $(\mathcal E, \mathcal A)$ be a corresponding tidy solution. By \Cref{cor:no-deg2}, the number of degree-2 vertices in $X$ is bounded by a polynomial of $|C|$ and $|T|$.
    Since all vertices of degree at least 3 are contained in~$C$, it only remains to bound the size of the set~$V_{\le 1}$ that contains all vertices of~$X \setminus C$ that have degree at most~$1$.

    From now on, we assume that $(\mathcal E, \mathcal A)$ minimizes the number of augmentation edges among all tidy solutions.
    Note that augmentation edges incident to vertices in~$V_{\le 1}$ are not needed to cover any ears of $V(G) \setminus X$.
    The minimality of $(\mathcal E, \mathcal A)$ thus implies that every vertex of $V_{\le 1}$ is only incident to augmentation edges of its own cluster and no such vertex is contained in a cycle of $(\mathcal E \cup \mathcal A)[T^{-1}(\mu)]$ for any cluster $\mu$.

A \emph{bridge} $B$ of $(\mathcal E, \mathcal A)$ is a maximal subset of $V_{\le 1}$ that is connected by augmentation edges.
    By the considerations above, all augmentation edges of $B$ belong to the same cluster $\mu$ and the augmentation edges connecting $B$ form a tree.
    The \emph{cluster attachments} $\att(B)$ of $B$ are the vertices of $X \setminus V_{\le 1}$ that belong to $\mu$ and are adjacent to a vertex of $B$ in $\mathcal E$ or $\mathcal A$; see \Cref{fig:bridgeA}.
    The \emph{extension} $\hat B$ of $B$ is the subgraph of $\mathcal E \cup \mathcal A$ formed by the vertex set $B \cup \att(B)$, all augmentation edges between vertices of $B$, and all (augmentation) edges between $B$ and $\att(B)$; see \Cref{fig:bridgeB}.
    Since all vertices of $\hat B$ belong to $\mu$ and $\hat B$ contains no edges between vertices of $\att(B)$, note that the minimality of the solution again implies that $\hat B$ is a tree.

    \begin{claim}
    \label{cl:bridge-leaves}
        Every leaf of $\hat B$ is contained in $X \setminus V_{\le 1}$.
    \end{claim}
    \begin{claimproof}[Proof of claim]
        Assume that there exists a leaf $v$ of $\hat B$ with $v \in V_{\le 1}$. 
        In particular, $v$ is not a cluster attachment of $B$ and $v$ is not adjacent to a vertex of $X \setminus V_{\le 1}$ that belongs to $\mu$ in $\mathcal E \cup \mathcal A$.  Thus~$v$ has aug-degree at most~$1$.  
        It follows from \Cref{rr:deg1} that $v$ is a pendant vertex of aug-degree~$1$.
        Moreover, since $v$ has aug-degree~1, its unique neighbor $v'$ in~$\mathcal E$ also belongs to $\mu$ by \Cref{rr:deg1}.
        Thus $v'$ is a cluster attachment of $B$ and also belongs to $\hat B$.
        But then $v$ has degree 2 in $\hat B$, a contradiction to the assumption that $v$ is a leaf of $\hat B$.
    \end{claimproof}

    We first bound the number of bridges in $(\mathcal E, \mathcal A)$.
    Since the number of clusters is $|T|$, it suffices to bound the number of bridges that belong to the same cluster $\mu$.
    Note that, if a tree has less than two leaves, every node of the tree is a leaf.
    Therefore, if the extension~$\hat B$ of a bridge $B$ has less than two leaves, it must be entirely contained in $X \setminus V_{\le 1}$ by \Cref{cl:bridge-leaves}, contradicting the definition of a bridge.
    Hence every bridge has at least two cluster attachments.
    Therefore, two bridges with the same set of cluster attachments immediately form a cycle in $(\mathcal E \cup \mathcal A)[T^{-1}(\mu)]$ that contains vertices in $V_{\le 1}$, which would contradict the minimality of the solution as argued above.
    Hence the number of distinct bridges that belong to cluster $\mu$ is bounded by the number of edges a planar graph with $|X \setminus V_{\le 1}|$ vertices can have, which is $O(|X\setminus V_{\le 1}|)$. 
    We therefore have $O(|X\setminus V_{\le 1}| \cdot |T|)$ bridges overall.
    Recall from above that $|X\setminus V_{\le 1}|$ is bounded by a polynomial of $|C|$ and $|T|$.

    Consider a bridge $B$ of $(\mathcal E, \mathcal A)$ and let $\mu$ denote the corresponding cluster.
    Since every leaf of $\hat B$ is contained in $X \setminus V_{\le 1}$ by \Cref{cl:bridge-leaves}, the number of leaves in $\hat B$, and thus also the number of its nodes with aug-degree at least 3, is linear in $|X \setminus V_{\le 1}|$.  Thus, removing all vertices of $\hat B$ that have aug-degree~$1$ or aug-degree at least~$3$ leaves $O(|X \setminus V_{\le 1}|)$ paths whose vertices all belong to~$V_{\le 1}$ and have aug-degree~$2$ in each bridge.
    By \Cref{rr:deg-1-paths}, each of these paths contains a polynomial number of pendant vertices.
    
    Thus, the isolated vertices in~$X \setminus C$ induce in~$(\mathcal E,\mathcal A)$ a polynomial number of induced augmentation paths.  We split each such path~$\pi$ at a vertex~$v$ of maximum level on~$\pi$ into two edge-disjoint subpaths that meet at~$v$.  We denote the set of these paths by~$\mathcal P'$ and we denote the vertices we used to split the paths by~$S$.  Note that $|\mathcal P'|$ is bounded by a polynomial of $|C|$ and $|T|$ and each path in~$\mathcal P'$ has an endpoint that is on the highest level of that path.

    Let~$\pi \in \mathcal P'$.  If~$\pi$ contains a vertex that is not a peak or valley, we can immediately shortcut and obtain a smaller instance with a tidy solution.  We may hence assume that~$\pi$ alternates between peaks and valleys.  Let~$v_1,\dots,v_t$ denote the vertices of~$\pi$ starting with the endpoint~$v_1$ on the highest level.  By definition, we have~$\ell(v_i) \le \ell(v_1)$ for all $1 \le i \le t$.  Moreover, if there is a vertex~$v_i$ with~$\ell(v_i) < \ell(v_{2})$, we can draw a curve~$\gamma$ below~$\pi$ that starts at~$v_1$ and closely follows~$\pi$ until it reaches the first valley~$v_j$ that lies on a level below~$\ell(v_2)$ so that~$\gamma$ intersects the levels~$\ell(v_1)$ and~$\ell(v_j)$ only with its endpoints.   By~\cite[Lemma 5]{BrucknerR21}, we can thus modify~$(\mathcal E,\mathcal A)$ in such a way that the edge~$v_1v_j$ can be inserted crossing-free and in a $y$-monotone fashion.  Then, removing all vertices of~$\pi[v_2,v_{j-1}]$ yields a smaller subinstance with a tidy solution.  We may hence assume that this does not occur and therefore~$v_2$ is a lowest vertex on~$\pi[v_2,v_t]$.  By an analogous argument, we can then prove that~$v_3$ is a highest vertex on~$\pi[v_2,v_t]$.  And thus, by induction, it follows that~$\pi$ has the shrinking property.

    Let now~$\mathcal P$ denote the vertex-disjoint paths obtained from~$\mathcal P'$ by removing the vertices in~$S$. 
    Then, $\mathcal P$ consists of a polynomial number of induced augmentation paths, each of which satisfies the shrinking property.  Let~$L$ be the set of vertices that are contained in a path from~$\mathcal P$.  Observe that all vertices in~$L$ have aug-degree~$2$ and thus the endpoints of each path in~$\mathcal P$ connect to two distinct vertices in~$X \setminus L$. Moreover, since we have bounded all vertices that are not contained in $L$ above, $|X \setminus L|$ is bounded by a polynomial. By applying the reduction of \Cref{rr:isolated}, it follows that $|L|$ is also bounded by a polynomial.  Therefore~$X = (X \setminus L) \cup L$ has polynomial size in $|C|$ and $|T|$.
\end{prooflater}

\subsection{Blueprints and Realizations}
It follows from \Cref{lm:boundedSize,lm:tidySubinstance} that an instance $\mathcal I$ of \yclp is a yes-instance if and only if it has a bounded-size subinstance that admits a tidy solution.
Note that we cannot enumerate all such bounded-size subinstances in FPT time.
Instead, we use so-called blueprints, which are instances of \yclp that coincide with $\mathcal I$ on its core graph.
Given a blueprint along with a solution for it, we show that we can test efficiently whether we can map the non-core vertices of the blueprint to corresponding vertices of $\mathcal I$ to obtain a tidy solution (of the subinstance formed by the image of this map).
This allows us to test the existence of a subinstance that admits a tidy solution by enumerating all bounded-size blueprints and their embeddings in FPT time and testing whether one such blueprint corresponds to a bounded-size subinstance of $\mathcal I$ that admits a tidy solution.

More formally, let $\mathcal I=(G, \ell, T)$ be an instance of \yclp with core vertex set~$C$.
A \emph{blueprint} of $\mathcal I$ is a yes-instance $\mathcal B = (H,\ell',T')$ of \yclp together with a solution $(\mathcal E, \mathcal A)$,~where
\begin{enumerate}[(B1)]
  \item $H$ is a graph that contains $G[C]$ as an induced subgraph,
  \item $\ell'$ contains no empty levels and for every pair $u, v \in C$, it is $\ell(u) \leq \ell(v)$ if and only if $\ell'(u) \leq \ell'(v)$, and
  \item the clusterings of $C$ given by $T$ and $T'$ are the same.
\end{enumerate}
These three properties express that $\mathcal I[C]$ can be obtained from $\mathcal B[C]$ by inserting empty levels.
The fact that no level of $\mathcal B$ is empty implies that the number of blueprints whose underlying graph $H$ is fixed is bounded by a function of the size of $H$ and $|T|$.
In particular, given a tidy solution of some subinstance $X$ of $\mathcal I$, removing all empty levels of $\mathcal I[X]$ yields a blueprint.

We now reverse this process and ask, for a given blueprint $\mathcal B$, whether it can be obtained in this way from a tidy solution $(\mathcal E_X, \mathcal A_X)$ of a subinstance $X$.
If this is the case, we say that $(\mathcal E_X, \mathcal A_X)$ realizes $\mathcal B$.
More formally, a \emph{realization} of $\mathcal B$ is a map $\Phi\colon V(H) \to V(G)$ such that
\begin{enumerate}[\label=(R1)]
  \item $\Phi$ is an isomorphism between $H$ and $G[\Phi(V(H))]$,
  \item $\Phi\vert_{C}$ is the identity function,
  \item\label{it:r2} for every pair $u, v \in V(H)$ it is  $\ell'(u) \leq \ell'(v)$ if and only if $\ell(\Phi(u)) \leq \ell(\Phi(v))$, and,
  \item\label{it:rcluster} for every $v \in V(H)$ it is $T'(v) = T(\Phi(v))$.
\end{enumerate}
In other words, $\Phi$ is an isomorphism between $H$ and the subgraph of $G$ induced by the image of~$H$ that additionally respects the order of the levels and the clustering.
For the solution~$(\mathcal E, \mathcal A)$ of $\mathcal B$, we let $\Phi((\mathcal E, \mathcal A))$ denote the embedding obtained from $(\mathcal E, \mathcal A)$ by renaming the vertices of $\mathcal B$ to vertices in $\mathcal I$ according to the isomorphism $\Phi$.
Since Property~(R\ref{it:r2}) guarantees that order of vertices along the $y$-axis remains consistent with (the leveling of) $\mathcal I$ and Property~(R\ref{it:rcluster}) guarantees that the clustering remains consistent, we obtain the following.

\begin{proposition}
\label{lm:realization}
  Let $\mathcal I$ be an instance of \yclp, let~$\mathcal B$ be a blueprint with solution $(\mathcal E, \mathcal A)$ and let~$\Phi$ be a realization of~$\mathcal B$. 
  Then $\Phi((\mathcal E, \mathcal A))$ is a solution of $\mathcal I[\Phi(V(B))]$.
\end{proposition}

A realization is \emph{tidy} if the solution $\Phi((\mathcal E, \mathcal A))$ is tidy.
The following lemma shows that we can efficiently test whether a blueprint admits a tidy realization.
Roughly speaking, we use dynamic programming to determine the combinations $b, l$ of levels of $\mathcal B$ and $\mathcal I$ that admit a \emph{tidy partial realization} that realizes all vertices of $\mathcal B$ up to level $b$ by vertices from $\mathcal I$ up to level~$l$ with a sufficient number of covering edges for the remaining~ears; see \Cref{fig:blueprints} for an example.

\begin{figure}
    \begin{subfigure}{.48\textwidth}
        \centering
        \includegraphics[page=1]{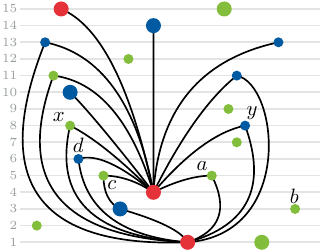}
        \subcaption{}
        \label{fig:blueprintA}
    \end{subfigure}
    \begin{subfigure}{.48\textwidth}
        \centering
        \includegraphics[page=2]{graphics/pdf/blueprints.pdf}
        \subcaption{}
        \label{fig:blueprintB}
    \end{subfigure}
    \newline
    \begin{subfigure}{.48\textwidth}
        \centering
        \includegraphics[page=3]{graphics/pdf/blueprints.pdf}
        \subcaption{}
        \label{fig:blueprintC}
    \end{subfigure}
    \begin{subfigure}{.48\textwidth}
        \centering
        \includegraphics[page=4]{graphics/pdf/blueprints.pdf}
        \subcaption{}
        \label{fig:blueprintD}
    \end{subfigure}
    \caption{
    \textbf{(a)} An instance $\mathcal I$ of of \yclp with three clusters. Large vertices represent vertices of the core.
    \textbf{(b)} A blueprint $\mathcal B$ of $\mathcal I$ along with a corresponding solution. Note that $\mathcal B$ contains all core vertices of $\mathcal I$, the remaining ``abstract'' vertices are illustrated as squares.
    \textbf{(c)} A tidy partial realization of the first six levels of $\mathcal B$ using vertices of the first ten levels of $\mathcal I$. The missing vertices of~$\mathcal I$ that are not part of the realization are greyed out.
    Note that the ears $x$ and $y$ (which are not part of the realization) are covered by the augmentation edges $e_1$ and $e_2$, respectively.
    \textbf{(d)} A realization of the entire blueprint, with missing vertices inserted into the solution.
    }
    \label{fig:blueprints}
\end{figure}

\begin{restatable}\restateref{lm:dp}{lemma}{lmDP}
\label{lm:dp}
  Let $\mathcal I = (G, \ell, T)$ be an instance of \yclp with core vertices $C$ and let $\mathcal B = (H,\ell',T')$ be a blueprint of $\mathcal I$ with solution $(\mathcal E, \mathcal A)$.  There is a polynomial-time algorithm to decide whether~$\mathcal B$ has a tidy realization.
\end{restatable}
\begin{prooflater}{proofDP}
    Let~$k$ denote the number of levels of~$\mathcal I$ and let~$l$ denote the number of levels of~$\mathcal B$.  For $0 \le i \le l$ let~$\mathcal B[i]$ denote the instance of \yclp consisting of only the vertices of level up to and including~$i$ in $\mathcal B$.

    Let~$0 \le k' \le k$ and~$0 \le l' \le l$.  A \emph{partial $(l',k')$-realization} of~$\mathcal B$ is a realization~$\Phi$ of~$\mathcal B[l']$ where $k'$ is the highest level of a vertex in the image of $\Phi$, that is $k' = \max(\{0\} \cup \{\ell(\Phi(v)) \mid v \in V(\mathcal B[l'])\})$.
    Let~$s$ be a level of~$\mathcal I$ and let~$uv$ be an augmentation edge in $\mathcal A$. We say that \emph{$uv$ crosses level $s$ with respect to $\Phi$} if (i) $\Phi(u)$ is below $s$ and (ii) $v$ is either not mapped by~$\Phi$ or $\Phi(v)$ is above $s$.
    A \emph{cover assignment} is an injective map that assigns ears in $\mathcal I$ to augmentation edges in~$\mathcal A$ such that each ear is covered by the augmentation edge to which it is assigned. 
    A partial $(l',k')$-realization~$\Phi$ is \emph{tidy} if for every level $s \leq k'$ there exists a cover assignment from the ears on level $s$ that are not in the image of $\Phi$ to the augmentation edges that cross $s$ with respect to $\Phi$.

    Consider the subproblem that asks whether~$\mathcal I$ has a tidy partial~$(l',k')$-realization of~$\mathcal B$.  We store this information in an~$(l+1) \times (k+1)$-table $T$ so that each entry~$T[l',k']$, where~$0 \le l' \le l$ and~$0 \le k' \le k$, stores whether there exists a tidy partial~$(l',k')$-realization.  Observe that the empty function is a tidy partial~$(0,0)$-realization and that a tidy partial $(l,k)$-realization is a tidy realization of~$\mathcal B$.

    Recall that the vertices of $\mathcal I$ and $\mathcal B$ coincide on the core $C$ and that every vertex of $V(\mathcal I) \setminus C$ is only adjacent to vertices of $C$.
    For a vertex $v \in V(\mathcal B) \setminus C$, a vertex $u$ of $\mathcal I$ is \emph{suitable} if $N(u) = N(v)$ and $T'(u) = T(v)$.
    For a vertex $v \in C$, the only vertex of $\mathcal I$ that is \emph{suitable} is $v$ itself.

    \begin{claim}
        There exists a tidy partial~$(l',k')$-realization of~$\mathcal B$ if and only if there exists an injective mapping~$\Phi_{l'}$ that maps vertices on level~$l'$ of the blueprint to suitable vertices on level~$k'$, and there exists a level~$k^\star$ of~$\mathcal I$ with~$0 \le k^\star < k'$ such that the following properties are satisfied:
        \begin{enumerate}[(i)]
        \item\label{it:i} there is a tidy partial~$(l'-1,k^\star)$-realization of~$\mathcal I$,
        \item\label{it:ii} for each level~$k^\star < s < k'$ there exists a cover assignment from the ears on level~$s$ to the augmentation edges~$uv$ of~$\mathcal A$ with~$\ell'(u) \le l'-1$ and~$\ell'(v) \ge l'$, and 
        \item\label{it:iii} for level~$k'$ there exists a cover assignment from the ears on level~$s$ that are not in the image of $\Phi_{l'}$ to the augmentation edges~$uv$ from~$\mathcal A$ with~$\ell'(u) \le l'-1$ and~$\ell'(v) > l'$.
        \end{enumerate}
    \end{claim}

    \begin{claimproof}[Proof of claim]
        Assume there exists a tidy partial~$(l', k')$-realization $\Phi$ of $\mathcal B$.
        Then the restriction $\Phi^\star$ of $\Phi$ to vertices of the levels $1, \dots, l'-1$ is a tidy partial~$(l'-1, k^\star)$-realization, where $k^\star$ is the highest level of a vertex in the image of $\Phi^\star$.
        Further, the restriction $\Phi$ to the vertices of level $l'$ is the desired injective mapping $\Phi_{l'}$.
        Note that $k^\star < k'$ by Property (R\ref{it:r2}) of the realization $\Phi$.
        Properties~(\ref{it:ii}) and~(\ref{it:iii}) follow directly from the definition of a tidy realization.

        For the other direction, assume we have have an injective mapping~$\Phi_{l'}$ that maps vertices on level~$l'$ of the blueprint to suitable vertices on level~$k'$, and level~$k^\star$ of~$\mathcal I$ with~$0 \le k^\star < k'$ such that Properties~(\ref{it:i})--(\ref{it:iii}) are satisfied.
        Let $\Phi'$ be a tidy partial~$(l'-1, k^\star)$-realization of $\mathcal I$.
        Let $\Phi$ be the unique function defined on $V(\mathcal B[l'])$ that coincides with $\Phi_{l'}$ and $\Phi'$ on their respective domain.
        Observe that, since $\Phi_{l'}$ maps the vertices on level $l'$ to suitable vertices on level $k'$, $\Phi$ is an $(l',k')$-realization of $\mathcal B$.
        We now show that $\Phi$ is tidy.
        To this end, we show for every level $s \leq k'$ of $\mathcal B$ that there exists a cover assignment $m_s$ from the ears on level~$s$ that are not in the image of $\Phi$ to the augmentation edges that cross $s$ with respect to $\Phi$. First consider the case $s \leq k^\star$.
        Note that every augmentation edge of $\mathcal A$ that crosses a level $s \leq k^\star$ of $\mathcal I$ with respect to $\Phi'$ also crosses level $s$ with respect to $\Phi$.
        Therefore, the existence of such a cover assignment $m_s$ follows from the fact that such a cover assignment exists for $\Phi'$.

        Now consider the case $k^\star < s < k'$.
        Then Property~(\ref{it:ii}) yields the desired mapping $m_s$.
        Similarly, for $s = k'$, we obtain $m_s$ from Property~(\ref{it:iii}).
        Hence $\Phi$ is tidy, which concludes the proof of the claim.
    \end{claimproof}

    Note that whether Property (\ref{it:iii}) from the claim holds is independent of the choice of~$\Phi_{l'}$, since the number of ears in each equivalence class that need to be mapped only depends on the vertices on level $k'$ in~$\mathcal I$ and the vertices on level~$l'$ in~$\mathcal B$.  As these numbers can be easily compared, the existence of such a mapping~$\Phi_l'$ can be checked efficiently.  Moreover, given~$\Phi_l'$, Properties~(\ref{it:ii}) and~(\ref{it:iii}) are simple assignment problems that can be tested in polynomial time.  Therefore, assuming that the entries of table~$T$ whose sum of indices is less then~$l'+k'$ have been correctly determined, this allows us to determine the entry~$T[l',k']$ in polynomial time.
\end{prooflater}

Using \Cref{lm:dp}, we are finally able to prove our main result of this section.
\begin{theorem}
\label{lem:yclp-fpt}
  \yclp is FPT w.r.t. the vertex cover number plus the number of clusters.
\end{theorem}
\begin{proof}
    Let $\mathcal I = (G, \ell, T)$ be an instance of \yclp with core vertices $C$.
    If $\mathcal I$ is a yes-instance, then it has a subinstance that admits a tidy solution and whose size does not exceed the bound $f(|C|, |T|)$ given in \Cref{lm:boundedSize}.
    As argued above, this determines a blueprint $\mathcal B$ such that (i) its graph has size at most $f(|C|, |T|)$ and (ii) it admits a tidy realization.
    Conversely, if an instance $\mathcal I$ has such a blueprint, it is a yes-instance by \Cref{lm:realization,lm:tidySubinstance}.
    
    Thus, to test whether an instance $\mathcal I$ of \yclp is a yes-instance, it suffices to test whether there exists a blueprint with these properties.
    Note that the number of embedded blueprints that satisfy property (i) is bounded, since blueprints contain no empty levels and the number of clustering-functions $T$ that represent distinct clusterings of a bounded-size graph is bounded.
    Given a blueprint $\mathcal B$ of $\mathcal I$ we use \Cref{lm:dp} to test whether $\mathcal B$ satisfies property~(ii). 
\end{proof}

\section{Conclusion}

We have extended previous hardness results for \yclp~\cite{fink_clustered_2024} and showed that the problem remains NP-complete even if all connected components have constant size and either the number of levels or the number of clusters is bounded.
This shows para-NP-hardness for most graph-structural parameters, but leaves the question whether the problem becomes tractable when parameterized by the vertex cover number.

In an effort towards answering this question, we gave an FPT-algorithm parameterized by the vertex cover number plus the number of clusters.
To handle the non-heredity of \yclp, we introduced the concept of solution normalizations and blueprints, which allow us to enumerate all sensible solutions in FPT time.
Interestingly, bounding the number of isolated vertices in solutions posed the biggest challenge.
Our main open question is whether the dependence on the number of clusters can be eliminated.
Here, a solution normalization to reduce the number of clusters alone would not suffice, since 
our dynamic program in \Cref{lm:dp} would have to additionally track suitable mappings from the clusters of the blueprint to the clusters of the input instance, which is significantly more difficult. It is also not clear whether the number of clusters can be bounded using conventional reduction rules.

\bibliography{references}

\end{document}